\documentclass[11pt,a4paper]{article}
\usepackage[pageanchor,hypertexnames=false]{hyperref}
\usepackage[dvipsnames]{xcolor}
\usepackage{amsmath,amssymb,amsbsy,amsthm,calc,cases,enumerate,url,framed,fancyhdr,lipsum,tikz,enumitem,natbib,etoolbox,bbm}
\usepackage[mathscr]{euscript}
\usepackage{xcolor}
\hypersetup{colorlinks,linkcolor={red!70!black},citecolor={blue!65!black},urlcolor={blue!70!black}}
\usepackage{graphicx}
\usepackage[justification=centering]{caption}

\usepackage{epstopdf}
\usepackage[linesnumbered,ruled]{algorithm2e}
\usepackage[font=small,skip=0pt]{caption}
\usepackage[all]{hypcap}

\definecolor{literal}{HTML}{e1d5e7}
\definecolor{proxy}{HTML}{d5e8d4}

\newcommand{\ind}{\perp\!\!\!\!\!\!\perp}

\def\p1{o_{\mathbb{P}}(1)}
\def\Op1{O_{\mathbb{P}}(1)}

\newcommand*\circled[1]{\tikz[baseline=(char.base)]{
            \node[shape=circle,draw,inner sep=1pt] (char) {#1};}}

\DeclareRobustCommand\sqr[1]{\tikz{\node[draw=white,fill=#1, fill opacity = 1,rectangle,minimum
width=0.25cm,minimum height=0.25cm,inner sep=0pt] at (0,0) {};}}

\definecolor{literal}{HTML}{e1d5e7}
\definecolor{proxy}{HTML}{d5e8d4}

\theoremstyle{plain} 
\newtheorem{theorem}{Theorem}
\newtheorem{proposition}{Proposition}
\newtheorem*{proposition*}{Proposition}
\newtheorem{lemma}{Lemma}
\newtheorem*{lemma*}{Lemma}

\newtheorem*{claim*}{Claim}
\newtheorem{definition}{Definition}
\newtheorem*{definition*}{Definition}
\newtheorem{assumption}{Assumption}

\newtheorem*{question*}{Question}

\newtheorem*{conjecture*}{Conjecture}

\theoremstyle{definition} 
\newtheorem{remark}{Remark}
\newtheorem*{remark*}{Remark}

\newtheorem*{aside*}{Aside}
\newtheorem{example}{Example}
\newtheorem*{example*}{Example}
\newtheorem*{aim*}{Aim}

\newtheorem*{observation*}{Observation}

\begin{document}

\title{The CATT SATT on the MATT: semiparametric inference for sample treatment effects on the treated}

\author{Andrew Yiu$^{1}$ \\
$^1$Department of Statistics, University of Oxford \\
\texttt{andrew.yiu@stats.ox.ac.uk}
}

\maketitle

\begin{abstract}
We study variants of the average treatment effect on the treated with population parameters replaced by their sample counterparts. For each estimand, we derive the limiting distribution with respect to a semiparametric efficient estimator of the population effect and provide guidance on variance estimation. Included in our analysis is the well-known sample average treatment effect on the treated, for which we obtain some unexpected results. Unlike the ordinary sample average treatment effect, we find that the asymptotic variance for  the sample average treatment effect on the treated is point-identified and consistently estimable, but it potentially exceeds that of the population estimand. To address this shortcoming, we propose a modification that yields a new estimand---the \textit{mixed average treatment effect on the treated}---which is always estimated more precisely than both the population and sample effects. We also introduce a second new estimand that arises from an alternative interpretation of the treatment effect on the treated with which all individuals are weighted by the propensity score.
\end{abstract}

\section{Introduction} \label{sec::intro}

In observational studies with no unmeasured confounding, causal effects are identified as structured combinations of marginal and conditional population parameters. If some---or even all---of these distributional components were replaced by their sample counterparts, we might expect the modified estimand to be inferable with greater precision, having stripped away layers of uncertainty in the population distribution \citep{Imbens24}. This increase in precision could be pivotal in settings with low statistical power, e.g., treatment effects that are small relative to the standard deviation of the outcomes \citep{Athey23}. 

For the average treatment effect, a complete characterization of such sample variants was provided by \citet{Imbens04}. Any asymptotically efficient estimator of the population estimand is also consistent and asymptotically normal for the sample average treatment effect but always with a smaller or equal asymptotic variance. This asymptotic variance is not point-identified, however, because it depends on the covariance between the two potential outcomes, requiring the use of a conservative variance estimator. An intermediate estimand---introduced by \citet{Abadie02}---averages the conditional treatment effect over the observed covariates rather than the population marginal covariate distribution. It is estimated with a precision bounded above and below by those of the sample and population estimands respectively.

Our contribution is to develop a similar characterization for the average treatment effect on the treated. It transpires that this theory is much richer for two reasons. First, the problem is asymmetric in the two treatment arms, which has surprising consequences for the commonly used \textit{sample average treatment effect on the treated} \citep{Robins88, Imbens04, Hartman15, Dorie19}. Despite involving both potential outcomes together, we show that its asymptotic variance is point-identified and consistently estimable. To the best of our knowledge, this provides the first asymptotically exact confidence interval procedure for the sample average treatment effect on the treated in observational studies, complementing a related finding by \citet{Sekhon21} for randomized clinical trials. Yet we also find that the asymptotic variance can exceed that of the population effect, rendering it arguably unsuitable as a sample variant.  We introduce an attractive alternative estimand---the \textit{mixed average treatment effect on the treated}---that avoids this problem.

Second, the average treatment effect on the treated has a dual interpretation. The \textit{literal} intepretation restricts attention to the individuals on treatment \citep[e.g.][]{Abadie02, Imbens04}, whereas estimands arising from the \textit{figurative} interpretation incorporate the entire population or sample, weighted by the propensity score. The asymptotic variances for these figurative estimands account for the variability in the treatment assignment, thus protecting against the possibility of an unrepresentative treated subsample. Included in this subfamily is another new estimand---the \textit{sample weighted average treatment effect on the treated}---that shares structural similarities with the sample average treatment effect.

\section{Set-up and point estimation} \label{sec::point}

Suppose that the variables $Z_{i}$ $(i = 1,\ldots,n)$ are independent and identically distributed replicates of $Z = (Y^{1}, Y^{0}, A, X)$ drawn from an unknown distribution $\mathbb{P}$. The real-valued or binary potential outcomes $Y^{1}$ and $Y^{0}$ correspond to treatment and control respectively, and $A$ is a binary variable that indicates the realized treatment assignment. We work in an observational setting: $X \in \mathbb{R}^{d}$ is a vector of covariates deemed sufficiently rich to adjust for confounding, and we require the propensity score $\pi(x) = \mathbb{P}(A = 1\mid X=x)$ to be bounded away from 0 and 1.
\begin{assumption} \label{ass::strong_ign_pos}
    (i) (Strong ignorability) Suppose $(Y^{0}, Y^{1}) \ind A \mid X$. (ii) (Positivity) There exists a number $0 < \delta < 1$ such that $\delta < \pi(x) < 1-\delta$ with $\mathbb{P}$-probability 1.
\end{assumption}
The data observed by the statistician are $(Y_{i}, A_{i},X_{i})$ $(i = 1,\ldots,n)$, where $Y_{i} = A_{i}Y^{1}_{i} + (1-A_{i})Y_{i}^{0}$. Under Assumption \ref{ass::strong_ign_pos}, the \textit{population average treatment effect on the treated} $\psi_{patt} = \mathbb{E}(Y^{1}-Y^{0} \mid A = 1)$ is identified by $\psi_{patt} = \mathbb{E}\{Y - \mu^{0}(X) \mid A = 1\}$, where $\mu^{a}(x) = \mathbb{E}(Y^{a} \mid X = x) = \mathbb{E}(Y \mid X=x, A=a)$. It should be understood that Assumption \ref{ass::strong_ign_pos} holds for the remainder of the paper without further statement.

We introduce some convenient notation. For any measurable function $f(z)$, let $P(f) = \int f(z)\,dP(z)$. In particular, $\mathbb{P}_{n}(f) = n^{-1}\sum_{i=1}^{n}f(Z_{i})$, where $\mathbb{P}_{n}$ is the empirical measure. Furthermore, let $L_{2}(\mathbb{P})$ denote the Hilbert space of all real-valued measurable functions $h$ with $\mathbb{P}[h^{2}] < \infty$ equipped with the inner product $\langle h_{1}, h_{2}\rangle = \mathbb{P}(h_{1}h_{2})$ and norm $\Vert h \Vert_{\mathbb{P}} = \surd \mathbb{P}(h^{2})$.

The analysis in this paper is based on an asymptotically efficient estimator of $\psi_{patt}$. For concreteness, we state the construction of an estimator using an estimating equation approach similar to \citet{Kennedy15} and provide sufficient conditions on the estimation of the nuisance parameters. The efficient influence function for $\psi_{patt}$ with unknown $\pi$ \citep{Hahn98} is
\begin{equation*}
    \dot{\psi} = \frac{A-\pi(X)}{\mathbb{P}(A)\{1-\pi(X)\}}\{Y - \mu^{0}(X)\} - \frac{A\psi_{patt}}{\mathbb{P}(A)}.
\end{equation*} 
Our estimator 
\begin{equation} \label{eqn::eff_est}
    \hat{\psi} = \mathbb{P}_{n}\left[\frac{A - \hat{\pi}(X)}{\mathbb{P}_{n}(A)\{1-\hat{\pi}(X)\}}\{Y - \hat{\mu}^{0}(X)\}\right]
\end{equation}
is defined by solving the empirical average of $\dot{\psi}$ after replacing $\pi$ and $\mu^{0}$ with user-specified estimators $\hat{\pi}$ and $\hat{\mu}^{0}$.
\begin{assumption} \label{ass::patt_eff}
    (i) Both $Y^{1}$ and $Y^{0}$ are square-integrable (ii) the sequences of estimators $\hat{\pi}$ and $\hat{\mu}^{0}$ each take values in fixed $\mathbb{P}$-Donsker classes (iii) $n^{1/2} \Vert \hat{\pi} - \pi \Vert_{\mathbb{P}} \Vert \hat{\mu}^{0} - \mu^{0} \Vert_{\mathbb{P}} = o_{\mathbb{P}}(1)$ (iv) there exist fixed positive constants $\varepsilon, C$ such that $\varepsilon < \hat{\pi} < 1-\varepsilon$ and $|Y - \hat{\mu}^{0}| < C$ with $\mathbb{P}$-probability 1.
\end{assumption}

\begin{remark}
    The Donsker condition in Assumption \ref{ass::patt_eff}(ii)---like the subsequent empirical process conditions in the paper---can be dropped if the construction of $\hat{\psi}$ is modified with \textit{sample-splitting} and \textit{cross-fitting} \citep[e.g.][]{Chernozhukov18, Hines22}. Assumption \ref{ass::patt_eff}(iii) is a \textit{rate double robustness} condition \citep{Rotnitzky21} that requires the combined convergence rate of $(\hat{\pi},\hat{\mu}^{0})$ to exceed $n^{-1/2}$.
\end{remark}

\begin{proposition} \label{prop::patt_eff}
    Under Assumption \ref{ass::patt_eff}, the estimator $\hat{\psi}$ in (\ref{eqn::eff_est}) admits the asymptotically linear expansion $n^{1/2}(\hat{\psi} - \psi_{patt}) = n^{1/2}\mathbb{P}_{n}(\dot{\psi}) + o_{\mathbb{P}}(1)$, and $\Vert\dot{\psi}\Vert_{\mathbb{P}} < \infty$.
\end{proposition}
We deduce that $n^{1/2}(\hat{\psi} - \psi_{patt})$ converges weakly to the normal distribution $\mathcal{N}(0, \Vert \dot{\psi}\Vert_{\mathbb{P}}^{2})$. The variance of $\hat{\psi}$ can be estimated by $n^{-1}\hat{V}_{patt}$, where $\hat{V}_{patt}$ is the sample variance of $\dot{\psi}$ after replacing $(\mathbb{P}(A), \pi, \mu^{0}, \psi_{patt})$ with $(\mathbb{P}_{n}(A), \hat{\pi}, \hat{\mu}^{0}, \hat{\psi})$.
\begin{proposition}
    Under Assumption \ref{ass::patt_eff}, $\hat{V}_{patt}$ converges to $\Vert \dot{\psi}\Vert_{\mathbb{P}}^{2}$ in $\mathbb{P}$-probability.
\end{proposition}

\begin{remark}
    Other approaches for constructing an asymptotically efficient estimator of $\psi_{patt}$ include \textit{regression imputation} \citep{Hahn98}, \textit{inverse probability weighting} \citep{Hirano03} and \textit{targeted maximum likelihood estimation} \citep[Chapter 8 of][]{vanderLaan11}. 
\end{remark}

In the subsequent sections, we will apply $\hat{\psi}$ to estimate different estimands. In each case, we obtain a weak convergence result of the form $n^{1/2}(\hat{\psi} - \psi^{*}) \rightarrow \mathcal{N}(0,V^{*})$; we abuse terminology slightly by saying that the estimand $\psi^{*}$ has asymptotic variance $V^{*}$. The following definition gives us a concise way of comparing precisions.
\begin{definition}
    For two estimands $\psi_{1}$ and $\psi_{2}$ with respective asymptotic variances $V_{1}$ and $V_{2}$, we say that $\psi_{1}$ is less conservative than $\psi_{2}$ if $V_{1} \leq V_{2}$ for all distributions $\mathbb{P}$. We denote this by $\psi_{1} \preceq \psi_{2}$; it is clear that the relation $\preceq$ is reflexive and transitive.
\end{definition}

Our theory is underpinned by orthogonal decompositions of efficient influence functions. We can write $\dot{\psi} = \dot{\psi}^{Y} + \dot{\psi}^{A}+\dot{\psi}^{X}$, where
\begin{align*}
    \dot{\psi}^{Y} &= \frac{Y - \mu^{A}(X)}{\mathbb{P}(A)}\left\{A - \frac{(1-A)\pi(X)}{1-\pi(X)}\right\}\\
    \dot{\psi}^{A} &= \frac{A - \pi(X)}{\mathbb{P}(A)}\{\mu^{1}(X) - \mu^{0}(X) - \psi_{patt}\}\\
    \dot{\psi}^{X} &= \frac{\pi(X)}{\mathbb{P}(A)}\{\mu^{1}(X) - \mu^{0}(X) - \psi_{patt}\}.
\end{align*}
Each component is the least favourable submodel score \citep[e.g. Chapter 8 of][]{vanderLaan11} corresponding to different factors of the observed data distribution: $\dot{\psi}^{Y}$, the conditional distribution of $Y$ given $(A,X)$; $\dot{\psi}^{A}$, the conditional distribution of $A$ given $X$; and $\dot{\psi}^{X}$, the marginal distribution of $X$. By mutual orthogonality, the variance of $\dot{\psi}$ decomposes into $\Vert \dot{\psi} \Vert^{2}_{\mathbb{P}} = \Vert \dot{\psi}^{Y} \Vert^{2}_{\mathbb{P}}+\Vert \dot{\psi}^{A} \Vert^{2}_{\mathbb{P}}+\Vert \dot{\psi}^{X} \Vert^{2}_{\mathbb{P}}$. We can interpret each variance component as the uncertainty in $\psi_{patt}$ induced by the corresponding factor of $\mathbb{P}$. 

\section{Estimands and inference}

\subsection{The literal and figurative interpretations} \label{sec::inter}

A complexity in defining sample variants of $\psi_{patt}$ arises from its dual interpretation. From the literal perspective, we restrict our attention to the subpopulation currently on treatment; that is, $\psi_{patt}$ quantifies the average effect of withholding treatment on the treated. This suggests defining a subfamily of \textit{literal estimands} that take the form of a simple average across the treated in the sample, such as the \textit{sample average treatment effect on the treated}
\begin{equation*}
    \psi_{satt} = \frac{\mathbb{P}_{n}[A(Y - Y^{0})]}{\mathbb{P}_{n}(A)}.
\end{equation*}

The alternative interpretation follows from writing $\psi_{patt}$ as
\begin{equation} \label{eqn::patt_prox}
    \psi_{patt} = \frac{\mathbb{E}[\pi(X)\{\mu^{1}(X) - \mu^{0}(X)\}]}{\mathbb{E}\{\pi(X)\}} = \frac{\mathbb{E}[\pi(X)\{Y^{1}-Y^{0}\}]}{\mathbb{E}\{\pi(X)\}},
\end{equation}
indicating that $\psi_{patt}$ is also the $\pi$-weighted average of the treatment effects across the whole population, putting more weight on individuals with a higher propensity of receiving treatment, and vice-versa. We call this the \textit{figurative} interpretation of the treatment effect on the treated, and \textit{figurative estimands} are accordingly defined as propensity-score-weighted averages  across the sample.

When we average across the infinite superpopulation, the causal effect defined by both interpretations coincide. Restricting to the sample, however, we obtain estimands with significantly different inferential properties. In the next subsections, we will study the limiting distributions of the estimator $\hat{\psi}$ defined in (\ref{eqn::eff_est}) for estimands in both subfamilies. Literal estimands condition on the sampling variability from the treatment assignment mechanism, so we do not need to account for the uncertainty in the unknown propensity score. Consequently, their asymptotic variances do not include the corresponding variance component $\Vert \dot{\psi}^{A} \Vert^{2}_{\mathbb{P}}$. This might be the appropriate option if, for example, we deem the treated in the sample to be typical of what we expect to see in the population. Figurative estimands are more conservative in this respect because they account for the variability in the treatment assignment, which provides protection against the risk of unrepresentative treated subsamples. We will also see that the two subfamilies behave differently when we consider estimands that involve the potential outcomes.

\subsection{Figurative estimands} \label{sec::figur}

Our first figurative estimand is the \textit{average covariate-conditional treatment effect on the treated} \citep{Yiu23}
\begin{equation} \label{eqn::actt}
    \psi_{actt} = \frac{\mathbb{P}_{n}[\pi(X)\{\mu^{1}(X) - \mu^{0}(X)\}]}{\mathbb{P}_{n}\{\pi(X)\}},
\end{equation}
which is defined by replacing the marginal distribution of $X$ in the middle expression of (\ref{eqn::patt_prox}) with its empirical counterpart $n^{-1}\sum_{i=1}^{n}\delta_{X_i}$. Since this estimand conditions on the observed covariates, we would expect to have $\psi_{actt} \preceq \psi_{patt}$ by removing the component $\Vert \dot{\psi}^{X}\Vert_{\mathbb{P}}^{2}$ from the asymptotic variance. This is confirmed by comparing Proposition \ref{prop::actt} below with Proposition \ref{prop::patt_eff}.

In contrast to the population effect, variance estimation for $\psi_{actt}$ requires estimating $\mu^{1}$. Fortunately, the assumptions on the estimator $\hat{\mu}^{1}$ are relatively mild; a Glivenko-Cantelli condition, rather than a Donsker condition, is sufficient, and a rate of $L_{2}$-convergence is unnecessary.
\begin{assumption} \label{ass::prox_eff}
    (i) The sequence of estimators $\hat{\mu}^{1}$ takes values in a fixed $\mathbb{P}$-Glivenko-Cantelli class (ii) $ \Vert \hat{\mu}^{1} - \mu^{1} \Vert_{\mathbb{P}} = o_{\mathbb{P}}(1)$ (iii) there exists a fixed positive constant $C$ such that $|Y - \hat{\mu}^{1}| < C$ with $\mathbb{P}$-probability 1.
\end{assumption}
\begin{proposition} \label{prop::actt}
    Under Assumption \ref{ass::patt_eff}, $n^{1/2}(\hat{\psi} - \psi_{actt}) \rightarrow \mathcal{N}(0,\Vert \dot{\psi}^{Y} + \dot{\psi}^{A}\Vert_{\mathbb{P}}^{2})$ in weak convergence. Let $\hat{V}_{actt}$ be the sample variance of $\dot{\psi}^{Y} + \dot{\psi}^{A}$ after replacing $(\mathbb{P}(A), \pi, \mu^{a}, \psi_{patt})$ with $(\mathbb{P}_{n}(A), \hat{\pi}, \hat{\mu}^{a}, \hat{\psi})$. If Assumption \ref{ass::prox_eff} also holds, then $\hat{V}_{actt}$ converges to $\Vert \dot{\psi}^{Y} + \dot{\psi}^{A}\Vert_{\mathbb{P}}^{2}$ in $\mathbb{P}$-probability.
\end{proposition}

The next estimand is based on the last expression in (\ref{eqn::patt_prox}). Starting from the previous estimand $\psi_{actt}$, the conditional expectations $(\mu^{0}, \mu^{1})$ are replaced by the potential outcomes $(Y^{0},Y^{1})$.
\begin{definition}
    The sample weighted average treatment effect on the treated is defined as
\begin{equation*}
    \psi_{swatt} = \frac{\mathbb{P}_{n}\{\pi(X)(Y^{1}-Y^{0})\}}{\mathbb{P}_{n}\{\pi(X)\}}.
\end{equation*}
\end{definition}

\begin{theorem} \label{theo::swatt}
    Under Assumption \ref{ass::patt_eff}, 
    \begin{equation*}
        n^{1/2}(\hat{\psi} - \psi_{swatt}) \rightarrow \mathcal{N}(0,\Vert \dot{\psi}^{Y} + \dot{\psi}^{A}\Vert_{\mathbb{P}}^{2} - \mathbb{P}(A)^{-2}\mathbb{E}\{\pi(X)^{2}\text{var}(Y^{1}-Y^{0} \mid X)\})
    \end{equation*}
    in weak convergence.
\end{theorem}
Inspecting the asymptotic variance in Theorem \ref{theo::swatt}, we see that $\psi_{swatt}$ is even less conservative than $\psi_{actt}$, so we have $\psi_{swatt} \preceq \psi_{actt} \preceq \psi_{patt}$. The difference is zero if and only if $\text{var}(Y^{1} - Y^{0} \mid X)=0$ with $\mathbb{P}$-probability 1. Unfortunately, the conditional variance of $Y^{1} - Y^{0}$ is not point-identified because we only observe one potential outcome per individual. If Assumption \ref{ass::prox_eff} holds, then $n^{-1}\hat{V}_{actt}$ provides a simple asymptotically conservative variance estimator. This is analogous to the approach advocated by \citet{Imbens04} for the sample average treatment effect. However, a sharper estimator is available if we are willing to undertake conditional variance estimation.

For each $a \in \{0,1\}$, let $\sigma^{2}_{a}(X) = \text{var}(Y^{a} \mid X)$. These conditional variances are identified by
\begin{equation*}
\sigma^{2}_{0}(X) = \mathbb{E}[\{Y - \mu^{0}(X)\}^{2} \mid A=0, X],\quad \sigma^{2}_{1}(X) = \mathbb{E}[\{Y - \mu^{1}(X)\}^{2} \mid A=1,X].
\end{equation*}
An application of the Cauchy-Schwarz inequality yields $|\text{cov}(Y^{1},Y^{0} \mid X)| \leq \sigma_{1}(X)\sigma_{0}(X)$, from which we deduce the lower bound $\mathbb{E}[\pi(X)^{2}\{\sigma_{1}(X) - \sigma_{0}(X)\}^{2}] \leq \mathbb{E}\{\pi(X)^{2}\text{var}(Y^{1}-Y^{0} \mid X)\}$. To estimate the lower bound, we require estimators $\hat{\sigma}_{a}(X)$ that satisfy the following conditions.
\begin{assumption} \label{ass::cond_var}
    For each $a \in \{0,1\}$, suppose: (i) the sequence of estimators $\hat{\sigma}_{a}$ takes values in a fixed uniformly bounded $\mathbb{P}$-Glivenko-Cantelli class; (ii) $ \Vert \hat{\sigma}_{a} - \sigma_{a} \Vert_{\mathbb{P}} = o_{\mathbb{P}}(1)$.
\end{assumption}
\begin{proposition}
    Under Assumptions \ref{ass::patt_eff} and \ref{ass::cond_var}, 
    \begin{equation*}
        \hat{V}_{\sigma} = \mathbb{P}_{n}(A)^{-2}\mathbb{P}_{n}[\hat{\pi}(X)^{2}\{\hat{\sigma}_{1}(X) - \hat{\sigma}_{0}(X)\}^{2}] \rightarrow \mathbb{P}(A)^{-2}\mathbb{E}[\pi(X)^{2}\{\sigma_{1}(X) - \sigma_{0}(X)\}^{2}]
    \end{equation*}
    in $\mathbb{P}$-probability.
\end{proposition}
Consequently, we can use $n^{-1}(\hat{V}_{actt} - \hat{V}_{\sigma})$ as an asymptotically conservative variance estimator for $\psi_{swatt}$. This is asymptotically sharper than just using $n^{-1}\hat{V}_{actt}$ unless $\sigma_{1}(X) = \sigma_{0}(X)$ with $\mathbb{P}$-probability 1. The sharpest possible bound follows from applying the Fr\'echet-Hoeffding upper bound to the conditional covariance. Estimating this is more involved, however, because it generally involves quantile regression for the potential outcome distributions. An exception is the case of binary outcomes, for which the sharpest bound takes a simple form and can be consistently estimated under Assumptions \ref{ass::patt_eff} and \ref{ass::prox_eff}.  Details are provided in the Appendix.

Providing an intuitive explanation for why $\psi_{swatt}$ is less conservative than $\psi_{actt}$ appears to be difficult. It is perhaps tempting to state that $\psi_{swatt}$ is conditioning on more of the variability in the data by substituting the potential outcomes $(Y^{0},Y^{1})$ for their conditional expectations $(\mu^{0}, \mu^{1})$. But this is fallacious because we never observe both potential outcomes together; this point is demonstrated by the results in the next subsection.

\subsection{Literal estimands} \label{sec::literal}

The \textit{conditional average treatment effect on the treated}
\begin{equation*} 
    \psi_{catt} = \frac{\mathbb{P}_{n}[A\{\mu^{1}(X) - \mu^{0}(X)\}]}{\mathbb{P}_{n}(A)}
\end{equation*}
was introduced by \citet{Abadie02}. It can be obtained by replacing the population joint distribution of $(A,X)$ in $\psi_{patt}$ with the empirical joint distribution $n^{-1}\sum_{i=1}^{n}\delta_{(A_{i},X_{i})}$. If we compare this to the covariate-conditional effect $\psi_{actt}$ in (\ref{eqn::actt}), which only replaces the marginal distribution of $X$, we would expect $\psi_{catt}$ to be less conservative, given that we have further removed the uncertainty regarding the conditional distribution of $A$ given $X$. This is confirmed by the following result.
\begin{proposition} \label{prop::catt}
    Under Assumption \ref{ass::patt_eff}, $n^{1/2}(\hat{\psi} - \psi_{catt}) \rightarrow \mathcal{N}(0,\Vert \dot{\psi}^{Y}\Vert_{\mathbb{P}}^{2})$ in weak convergence. Let $\hat{V}_{catt}$ be the sample variance of $\dot{\psi}^{Y}$ after replacing $(\mathbb{P}(A), \pi, \mu^{a})$ with $(\mathbb{P}_{n}(A), \hat{\pi}, \hat{\mu}^{a})$. If Assumption \ref{ass::prox_eff} also holds, then $\hat{V}_{catt}$ converges to $\Vert \dot{\psi}^{Y}\Vert_{\mathbb{P}}^{2}$ in $\mathbb{P}$-probability.
\end{proposition}

Recall that the population effect is identified by $\psi_{patt} = \mathbb{E}(Y \mid A = 1) - \tau$, where $\tau = \mathbb{E}(Y^{0} \mid A=1) = \mathbb{E}\{\mu^{0}(X) \mid A = 1\}$. The remaining two estimands replace $\mathbb{E}(Y \mid A = 1)$ with the sample average $\mathbb{P}_{n}(AY)/\mathbb{P}_{n}(A)$, which is completely determined by the observed data. Thus, our analysis now revolves around the functional $\tau$; its efficient influence function has the orthogonal decomposition
\begin{align*}
    \dot{\tau} &= \dot{\tau}^{Y} + \dot{\tau}^{A}+\dot{\tau}^{X}\\
    &= \frac{\{Y - \mu^{0}(X)\}(1-A)\pi(X)}{\mathbb{P}(A)\{1-\pi(X)\}} + \frac{\{A - \pi(X)\} \{\mu^{0}(X)-\tau\}}{\mathbb{P}(A)} +\frac{\pi(X)\{\mu^{0}(X) - \tau\}}{\mathbb{P}(A)}.
\end{align*}

We return to the sample treatment effect on the treated $\psi_{satt} = \mathbb{P}_{n}[A(Y - Y^{0})] /\mathbb{P}_{n}(A)$ discussed earlier.
Perhaps surprisingly, despite involving both potential outcomes together, the asymptotic variance of $\psi_{satt}$ is point-identified and can be consistently estimated.
\begin{theorem} \label{theo::satt}
    Suppose Assumption \ref{ass::patt_eff} holds. Then
    \begin{equation*}
        n^{1/2}(\hat{\psi} - \psi_{satt}) \rightarrow \mathcal{N}(0,\Vert \dot{\tau}^{Y}\Vert_{\mathbb{P}}^{2} + \mathbb{P}(A)^{-1}\mathbb{E}\{\text{var}(Y^{0} \mid X) \mid A= 1\})
    \end{equation*}
    in weak convergence and
    \begin{equation*}
        \hat{V}_{satt} = \mathbb{P}_{n}\left[\hat{\pi}(X)\frac{(1-A)}{\{1-\hat{\pi}(X)\}^{2}}\left\{\frac{Y - \hat{\mu}^{0}(X)}{\mathbb{P}_{n}(A)}\right\}^{2}\right]
    \end{equation*}
    is a consistent estimator of the asymptotic variance.
\end{theorem}
Remarkably, it can be deduced from Theorem \ref{theo::satt} that the asymptotic variance of $\psi_{satt}$ potentially exceeds that of the population effect $\psi_{patt}$, as illustrated by the following example.
\begin{example} \label{exa::satt}
    Suppose the treatment effect is homogeneous; that is, $\mu^{1}-\mu^{0}$ is constant. Then $\psi_{patt} = \mu^{1}-\mu^{0}$, and it follows that $\dot{\psi}^{A} = \dot{\psi}^{X} = 0$.  It is straightforward to show that the asymptotic variance of $\psi_{patt}$ is $\Vert\dot{\tau}^{Y}\Vert_{\mathbb{P}}^{2} + \mathbb{P}(A)^{-1}\mathbb{E}\{\text{var}(Y^{1} \mid X) \mid A= 1\}$, which is strictly less than the asymptotic variance of $\psi_{satt}$ if and only if $\mathbb{E}\{\text{var}(Y^{1} \mid X)  - \text{var}(Y^{0} \mid X)\mid A= 1\} < 0$.
\end{example}
\begin{remark}
    A similar argument yields $\mathbb{E}\{\text{var}(Y^{1} \mid X)  - \text{var}(Y^{0} \mid X)\mid A= 1\} \geq 0$ as a general sufficient condition for the asymptotic variance of $\psi_{satt}$ to be upper-bounded by that of $\psi_{patt}$.
\end{remark}
The above phenomenon highlights that the inference for $\psi_{satt}$ requires accounting for a source of uncertainty that is absent when inferring $\psi_{patt}$, namely, the variance of $Y^{0}$ that is unexplained by $X$ when $A=1$. Besides the undesirable possibility of obtaining a wider confidence interval for the sample effect, this is problematic if we wish to interpret $\psi_{satt}$ as a distillation of the information contained in the observed data about $\psi_{patt}$; see \citet{Imbens04} for a related discussion. To resolve this issue, we introduce a new estimand that replaces the awkward $AY^{0}$ term in $\psi_{satt}$ with $A\mu^{0}$.
\begin{definition}
    The mixed average treatment effect on the treated is defined as
\begin{equation*}
    \psi_{matt} = \frac{\mathbb{P}_{n}[A\{Y-\mu^{0}(X)\}]}{\mathbb{P}_{n}(A)}.
\end{equation*}
\end{definition}
\begin{proposition} \label{prop::matt}
    Suppose Assumption \ref{ass::patt_eff} holds, and let $\hat{V}_{matt}$ be the sample variance of $\dot{\tau}^{Y}$ after replacing $(\mathbb{P}(A), \pi, \mu^{0})$ with $(\mathbb{P}_{n}(A), \hat{\pi}, \hat{\mu}^{0})$. Then $n^{1/2}(\hat{\psi} - \psi_{matt}) \rightarrow \mathcal{N}(0,\Vert \dot{\tau}^{Y}\Vert_{\mathbb{P}}^{2})$ in weak convergence, and $\hat{V}_{matt}$ converges to $\Vert \dot{\tau}^{Y}\Vert_{\mathbb{P}}^{2}$ in $\mathbb{P}$-probability. Moreover, we have $\Vert \dot{\tau}^{Y}\Vert_{\mathbb{P}} \leq \Vert \dot{\psi}^{Y}\Vert_{\mathbb{P}}$, so $\psi_{matt} \preceq \psi_{catt} $.
\end{proposition}
Combined with our earlier results in Propositions \ref{prop::actt} and \ref{prop::catt}, we deduce that $\psi_{matt}$ is less conservative than both $\psi_{satt}$ and $\psi_{patt}$. This makes $\psi_{matt}$ an attractive alternative to $\psi_{satt}$ that retains a similar literal interpretation. An advantage over $\psi_{catt}$ and the figurative estimands is that estimation of $\mu^{1}$ is not required.

\begin{figure}
\centering
\includegraphics[width=0.75\textwidth]{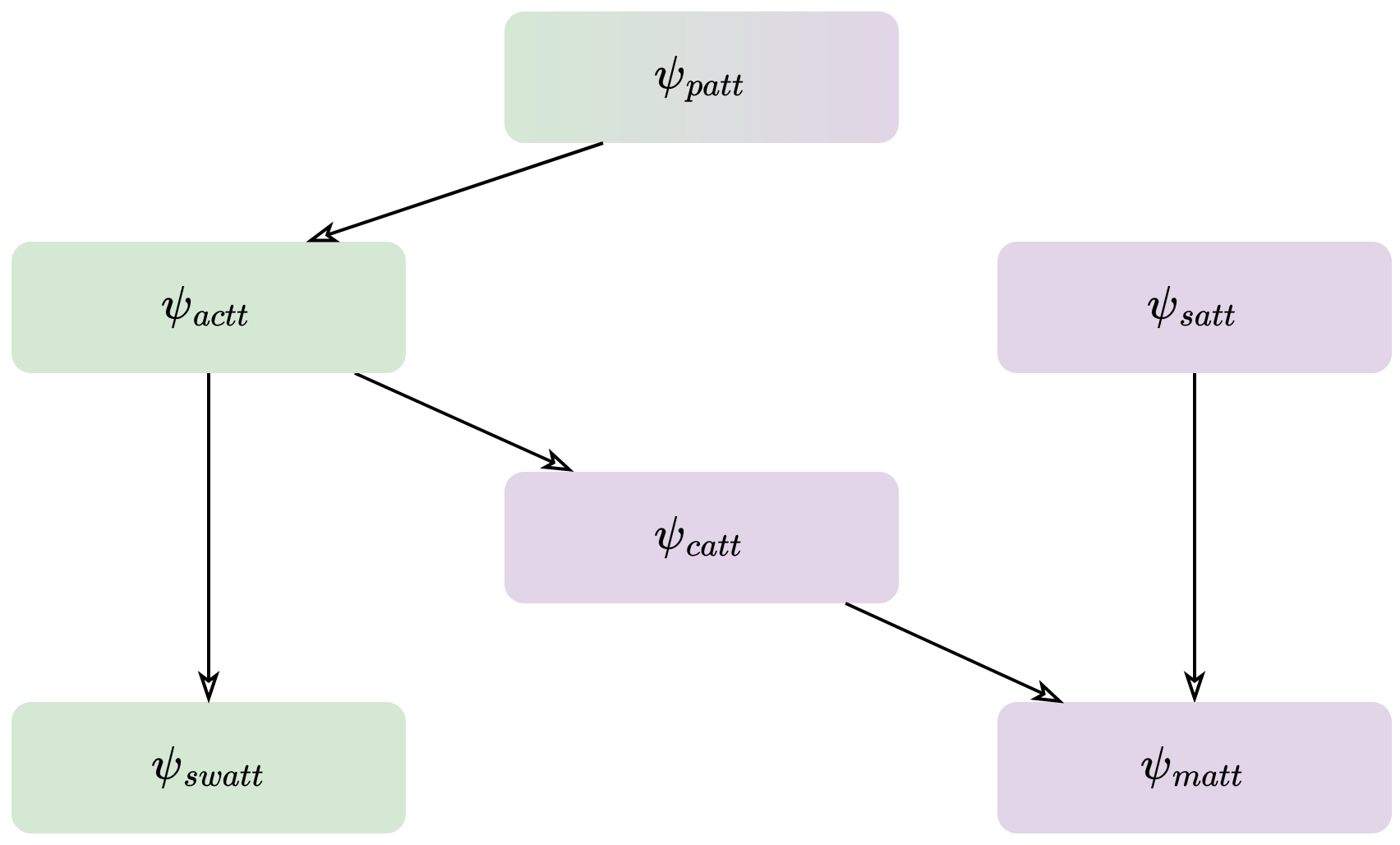}
\vspace{10mm}
\caption{Hasse diagram describing the partial ordering on the estimands;\\ \sqr{proxy} figurative estimands; \sqr{literal} literal estimands; $\psi_{1} \leftarrow \psi_{2}$ means that $\psi_{1} \preceq \psi_{2}$. \label{fig1}}

\end{figure}

\begin{table}[h!]
\begin{center}
\resizebox{.8\textwidth}{!}
{	\begin{tabular}{|c|l|}
	\hline
		Estimand &  Asymptotic variance  \\ 
		\hline
		$\psi_{patt}$ & $\Vert \dot{\psi}^{Y} \Vert^{2}_{\mathbb{P}}+\Vert \dot{\psi}^{A} \Vert^{2}_{\mathbb{P}}+\Vert \dot{\psi}^{X} \Vert^{2}_{\mathbb{P}}$\\ 
		$\psi_{actt}$  & $\Vert \dot{\psi}^{Y} \Vert^{2}_{\mathbb{P}}+\Vert \dot{\psi}^{A} \Vert^{2}_{\mathbb{P}}$\\ 
        $\psi_{swatt}$  & $\Vert \dot{\psi}^{Y} \Vert^{2}_{\mathbb{P}}+\Vert \dot{\psi}^{A} \Vert^{2}_{\mathbb{P}}- \mathbb{P}(A)^{-2}\mathbb{E}\{\pi(X)^{2}\text{var}(Y^{1}-Y^{0} \mid X)\}$\\ 
        $\psi_{catt}$ & $\Vert \dot{\psi}^{Y} \Vert^{2}_{\mathbb{P}}$\\ 
        $\psi_{satt}$ & $\Vert \dot{\tau}^{Y} \Vert^{2}_{\mathbb{P}}+ \mathbb{P}(A)^{-1}\mathbb{E}\{\text{var}(Y^{0} \mid X) \mid A= 1\}$\\ 
        $\psi_{matt}$  & $\Vert \dot{\tau}^{Y} \Vert^{2}_{\mathbb{P}}$\\ 
	\hline
	\end{tabular}}
 \end{center}
 \caption{Comparison of asymptotic variances.\label{tablelabel}}
	
\end{table}

Figure \ref{fig1} and Table \ref{tablelabel} summarize our results, describing the partial ordering on our estimands defined by the relation $\preceq$. As discussed in Section \ref{sec::inter}, $\psi_{patt}$ is the unique estimand that enjoys both the literal and figurative interpretations. There exist further estimands that belong in neither class, but we have deemed them to be less practically relevant. Details can be found in the Appendix, which also includes examples to justify why the diagram contains no directed paths between $\psi_{swatt}$ and any of $\{\psi_{catt}, \psi_{satt},\psi_{matt}\}$. We emphasize again that our results are relative to an asymptotically efficient estimator for $\psi_{patt}$. For future work, it would be of interest to either verify that $\hat{\psi}$ is efficient for the sample variants, or to show that there exist more efficient estimators. 

\section*{Acknowledgements}
The author receives funding from Novo Nordisk and thanks Edwin Fong for helpful suggestions that improved the paper.

\bibliographystyle{abbrvnat}
\bibliography{reference}

\appendix

\section{Proofs of results in the main text}

\subsection{Proof of Proposition 1}

By assumption, there exist fixed constants $\delta,\varepsilon > 0$ such that
\begin{equation*}
    \delta < \pi < 1-\delta, \quad \varepsilon < \hat{\pi} < 1-\varepsilon
\end{equation*}
with $\mathbb{P}$-probability 1. Without loss of generality, we can redefine $\delta$ as $\min\{\delta, \varepsilon\}$ so that it can be used to bound both $\pi$ and $\hat{\pi}$ for the sake of convenience.

We start by showing that $\Vert \dot{\psi} \Vert^{2}_{\mathbb{P}} < \infty$, i.e. the semiparametric efficiency bound for $\psi_{patt}$ is finite:
\begin{align*}
    \Vert \dot{\psi} \Vert^{2}_{\mathbb{P}} &= \mathbb{P}[A]^{-2}\mathbb{P}\left[\pi(X)\sigma_{1}^{2}(X)+\frac{\pi(X)^{2}}{1-\pi(X)}\sigma_{0}^{2}(X) + \pi(X)^{2}\{\mu^{1}(X) - \mu^{0}(X) - \psi_{patt}\}^{2}\right]\\
    &\leq \mathbb{P}[A]^{-2}\left(\mathbb{P}[\sigma^{2}_{1}(X)] + \frac{1}{\delta}\mathbb{P}[\sigma^{2}_{0}(X)] + 3\mathbb{P}[\pi(X)^{2}\{\mu^{1}(X)^{2} + \mu^{0}(X)^{2} + \psi^{2}_{patt}\}] \right)\\
    &\lesssim \mathbb{P}[\sigma^{2}_{1}(X)] + \mathbb{P}[\sigma^{2}_{0}(X)] + \mathbb{P}[\mu^{1}(X)^{2}] + \mathbb{P}[\mu^{0}(X)^{2}] + \psi_{patt}^{2}\\
    &= \Vert Y^{1}\Vert^{2}_{\mathbb{P}} + \Vert Y^{0}\Vert^{2}_{\mathbb{P}} + \psi_{patt}^{2} \\
    &< \infty,
\end{align*}
where the boundedness on the final line follows from the assumed square-integrability of $Y^{1}$ and $Y^{0}$.

We will use $\mathbb{G}_{n} = n^{1/2}(\mathbb{P}_{n}-\mathbb{P})$ to denote the empirical process. It is helpful to define
\begin{align*}
    \hat{h}(Z) &= \left(\frac{A-\hat{\pi}(X)}{1-\hat{\pi}(X)}\right)\{Y-\hat{\mu}^{0}(X)\}-A\psi_{patt} \\
    h(Z) &= \left(\frac{A-\pi(X)}{1-\pi(X)}\right)\{Y-\mu^{0}(X)\}-A\psi_{patt} = \mathbb{P}[A]\dot{\psi}.
\end{align*}
Then by the definition of $\hat{\psi}$, we have
\begin{equation*}
    n^{1/2}(\hat{\psi} - \psi_{patt}) = n^{1/2}\frac{\mathbb{P}_{n}[\hat{h}(Z)]}{\mathbb{P}_{n}[A]}.
\end{equation*}
Consider the decomposition
\begin{equation*}
    n^{1/2}\mathbb{P}_{n}[\hat{h}] = \mathbb{G}_{n}[h] + \underbrace{\mathbb{G}_{n}[\hat{h} - h]}_{\circled{1}} + \underbrace{n^{1/2}\mathbb{P}[\hat{h}]}_{\circled{2}},
\end{equation*}
where we have used the fact that $\mathbb{P}[h] = 0$. We will show that terms \circled{1} and \circled{2} are both $o_{\mathbb{P}}(1)$.

For term \circled{1},
\begin{align*}
    \hat{h}(Z) - h(Z) =& \left[\frac{A-\hat{\pi}(X)}{1-\hat{\pi}(X)}\right]\{Y-\hat{\mu}^{0}(X)\} - \left[\frac{A-\pi(X)}{1-\pi(X)}\right]\{Y-\mu^{0}(X)\}\\
    =& \left(1-\frac{1-A}{1-\pi(X)}\right)\{\mu^{0}(X)-\hat{\mu}^{0}(X)\} \\
    &+ (1-A)\left[\frac{1}{1-\pi(X)}-\frac{1}{1-\hat{\pi}(X)}\right]\{Y - \hat{\mu}^{0}(X)\} \\
    =& \left(1-\frac{1-A}{1-\hat{\pi}(X)}\right)\{\mu^{0}(X)-\hat{\mu}^{0}(X)\} \\
    &+ (1-A)\left[\frac{\pi(X)-\hat{\pi}(X)}{(1-\pi(X))(1-\hat{\pi}(X))}\right]\{Y - \hat{\mu}^{0}(X)\}
\end{align*}
so
\begin{equation*}
    \Vert \hat{h} - h \Vert_{\mathbb{P}} \leq \left(1 + \frac{1}{\delta}\right)\Vert \mu^{0} - \hat{\mu}^{0}\Vert_{\mathbb{P}} + \left(\frac{C}{\delta^{2}}\right)\Vert \hat{\pi} - \pi \Vert_{\mathbb{P}},
\end{equation*}
which is $o_{\mathbb{P}}(1)$ by the assumptions of $L_{2}$-convergence of $\mu^{0}$ and $\pi$. Using Assumption \ref{ass::patt_eff} and repeated applications of Lemma \ref{lem::perm_Dons}, the sequence $\hat{h}$ takes values in a fixed $\mathbb{P}$-Donsker class. Then we can apply Lemma 19.24 of \citet{vanderVaart98} to obtain $\mathbb{G}_{n}[\hat{h} - h] = o_{\mathbb{P}}(1)$ as required.

For term \circled{2}, we have
\begin{align*}
    \mathbb{P}[\hat{h}(Z)] =& \mathbb{P}\left[\left(\frac{A-\hat{\pi}(X)}{1-\hat{\pi}(X)}\right)\{Y-\hat{\mu}^{0}(X)\}-\pi(X)\{\mu^{1}(X)-\mu^{0}(X)\}\right]\\
    =& \mathbb{P}\left[\frac{\pi(X)\mu^{1}(X) - \pi(X)\hat{\mu}^{0}(X) - \hat{\pi}(X)[\pi(X)\mu^{1}(X) + (1-\pi(X))\mu^{0}(X)] + \hat{\pi}(X)\hat{\mu}^{0}(X)}{1-\hat{\pi}(X)}\right]\\
    &-\mathbb{P}\left[\pi(X)\{\mu^{1}(X)-\mu^{0}(X)\}\right]\\
    =& \mathbb{P}\left[\frac{ - \pi(X)\hat{\mu}^{0}(X) - \hat{\pi}(X) (1-\pi(X))\mu^{0}(X) + \hat{\pi}(X)\hat{\mu}^{0}(X) + \pi(X)\mu^{0}(X)(1-\hat{\pi}(X))}{1-\hat{\pi}(X)}\right]\\
    =& \mathbb{P}\left[\frac{ - \pi(X)\hat{\mu}^{0}(X) - \hat{\pi}(X)\mu^{0}(X) + \hat{\pi}(X)\hat{\mu}^{0}(X) + \pi(X)\mu^{0}(X))}{1-\hat{\pi}(X)}\right]\\
    =& \mathbb{P}\left[\frac{(\hat{\pi}(X)-\pi(X))(\hat{\mu}^{0}(X)-\mu^{0}(X))}{1-\hat{\pi}(X)}\right]\\
    \leq & \frac{1}{\delta}\Vert \hat{\pi} - \pi \Vert_{\mathbb{P}} \Vert \mu^{0} - \hat{\mu}^{0}\Vert_{\mathbb{P}}\\
    =& o_{\mathbb{P}}(n^{-1/2}),
\end{align*}
where the inequality on the penultimate line is due to Cauchy-Schwarz and the assumed bounding on $\hat{\pi}$. 

Thus, we have shown that 
\begin{equation*}
    n^{1/2}(\hat{\psi} - \psi_{patt}) = \frac{\mathbb{G}_{n}[h] + o_{\mathbb{P}}(1)}{\mathbb{P}_{n}[A]}.
\end{equation*}
Since $\mathbb{P}[\dot{\psi}^{2}] < \infty$, we also have $\mathbb{P}[h^{2}] < \infty$, and the central limit theorem yields $\mathbb{G}_{n}[h] = O_{\mathbb{P}}(1)$. So Lemma \ref{lem::swap} implies that
\begin{equation*}
    n^{1/2}(\hat{\psi} - \psi_{patt}) = \frac{\mathbb{G}_{n}[h] + o_{\mathbb{P}}(1)}{\mathbb{P}[A]} + o_{\mathbb{P}}(1) = \mathbb{G}_{n}[\dot{\psi}] + o_{\mathbb{P}}(1).
\end{equation*}  

\subsection{Proof of Proposition 2}

Let
\begin{equation*}
    \Tilde{h}(Z) = \left(\frac{A-\hat{\pi}(X)}{1-\hat{\pi}(X)}\right)\{Y-\hat{\mu}^{0}(X)\}-A\hat{\psi},
\end{equation*}
which differs from $\hat{h}$ in the proof of Proposition 1 only in replacing $\psi_{patt}$ with $\hat{\psi}$. Since $\hat{\psi}$ takes values in a fixed bounded set, we deduce from Lemma \ref{lem::perm_Dons} and the proof of Proposition 1 that $\Tilde{h}$ also takes values in a fixed $\mathbb{P}$-Donsker class, which is, moreover, a $\mathbb{P}$-Glivenko-Cantelli class. Furthermore, Lemma \ref{lem::perm_GC} implies that $\Tilde{h}^{2}$ takes values in a $\mathbb{P}$-Glivenko-Cantelli class. Thus,
\begin{equation*}
    \mathbb{P}_{n}[\Tilde{h}^{2}] = \mathbb{P}[\Tilde{h}^{2}] + o_{\mathbb{P}}(1).
\end{equation*}
We will show that $\Vert \Tilde{h} - h\Vert_{\mathbb{P}} = o_{\mathbb{P}}(1)$, which suffices to establish $\mathbb{P}[\Tilde{h}^{2}] = \mathbb{P}[h^{2}] + o_{\mathbb{P}}(1)$ by Lemma \ref{lem::l2}. As mentioned above, we have $\Tilde{h} = \hat{h} + A(\psi_{patt} - \hat{\psi})$. Hence, Minkowski's inequality yields
\begin{equation*}
    \Vert \Tilde{h} - h\Vert_{\mathbb{P}} \leq \Vert \hat{h} - h\Vert_{\mathbb{P}} + \Vert A(\hat{\psi} - \psi_{patt})\Vert_{\mathbb{P}}.
\end{equation*}
The first term on the right-hand side of the inequality was already shown to be $o_{\mathbb{P}}(1)$ in the proof of Proposition 1. And the second term is upper-bounded by $|\hat{\psi} - \psi_{patt}|$, which is $o_{\mathbb{P}}(1)$ by Proposition 1. A similar argument yields $\mathbb{P}_{n}[\Tilde{h}] = \mathbb{P}[h] + o_{\mathbb{P}}(1) = o_{\mathbb{P}}(1)$.

Finally, we can write
\begin{equation*}
    \hat{V}_{patt} = \frac{\mathbb{P}_{n}[\Tilde{h}^{2}] - \mathbb{P}_{n}[\Tilde{h}]^{2}}{\mathbb{P}_{n}[A]^{2}}.
\end{equation*}
Two applications of Lemma \ref{lem::swap} yields
\begin{equation*}
    \hat{V}_{patt} = \frac{\mathbb{P}[h^{2}]}{\mathbb{P}[A]^{2}}+ \p1 = \mathbb{P}[\dot{\psi}^{2}] + \p1.
\end{equation*}

\subsection{Proof of Proposition 3}

Using Proposition 1, we can start with $n^{1/2}(\hat{\psi} - \psi_{patt}) = n^{1/2}\mathbb{P}_{n}[\dot{\psi}] + o_{\mathbb{P}}(1)$, so
\begin{equation*}
    n^{1/2}(\hat{\psi} - \psi_{actt}) = n^{1/2}(\psi_{patt} + \mathbb{P}_{n}[\dot{\psi}^{Y} + \dot{\psi}^{A}] + \mathbb{P}_{n}[\dot{\psi}^{X}] - \psi_{actt}) + \p1.
\end{equation*}
Thus, it suffices to show that $n^{1/2}(\psi_{patt} + \mathbb{P}_{n}[\dot{\psi}^{X}] - \psi_{actt}) = \p1$ since the remaining term $n^{1/2}\mathbb{P}_{n}[\dot{\psi}^{Y} + \dot{\psi}^{A}]$ converges weakly to the desired normal distribution. Some simple manipulations yield
\begin{equation*}
    n^{1/2}(\psi_{patt} + \mathbb{P}_{n}[\dot{\psi}^{X}] - \psi_{actt}) = \frac{n^{1/2}(\mathbb{P}[A]-\mathbb{P}_{n}[\pi])}{\mathbb{P}[A]}\left\{\psi_{patt} - \frac{\mathbb{P}_{n}[\pi(\mu^{1} - \mu^{0})]}{\mathbb{P}_{n}[\pi]}\right\}.
\end{equation*}
Since $\mathbb{P}[A] = \mathbb{P}[\pi]$, the central limit theorem implies that $n^{1/2}(\mathbb{P}[A]-\mathbb{P}_{n}[\pi])$ is $\Op1$. By the weak law of large numbers,
\begin{align*}
    \mathbb{P}_{n}[\pi(\mu^{1} - \mu^{0})] &\rightarrow \mathbb{P}[\pi(\mu^{1} - \mu^{0})] \\
    \mathbb{P}_{n}[\pi] &\rightarrow \mathbb{P}[\pi]
\end{align*}
in $\mathbb{P}$-probability. Combining these with Slutsky's lemma and using the figurative representation of $\psi_{patt}$, we deduce that the term in curly brackets is $\p1$. Using Slutsky's lemma again, we have that $n^{1/2}(\psi_{patt} + \mathbb{P}_{n}[\dot{\psi}^{X}] - \psi_{actt})$ converges weakly to zero. Since zero is a constant, we also have convergence in $\mathbb{P}$-probability.

Now we proceed to variance estimation. Define
\begin{align*}
    \Tilde{h}_{Y} &= (Y - \hat{\mu}^{A})\left\{A - \frac{(1-A)\hat{\pi}}{1-\hat{\pi}}\right\}\\
    \Tilde{h}_{A} &= (A - \hat{\pi})\{\hat{\mu}^{1} - \hat{\mu}^{0} - \hat{\psi}\}\\
    \Tilde{h}_{X} &= \hat{\pi}\{\hat{\mu}^{1} - \hat{\mu}^{0} - \hat{\psi}\}.
\end{align*}
It is clear that the above expressions sum to $\Tilde{h}$ as defined in the proof of Proposition 2. The functions $h_{Y}$, $h_{A}$ and $h_{X}$ are defined similarly, replacing the estimators of the parameters with their true values. Assumptions \ref{ass::patt_eff} and \ref{ass::prox_eff}, along with repeated applications of Lemma \ref{lem::perm_GC}, imply that each of $(\Tilde{h}_{Y}, \Tilde{h}_{A}, \Tilde{h}_{X})$ takes values in a fixed $\mathbb{P}$-Glivenko-Cantelli class. This is also the case for $\Tilde{h}_{Y} + \Tilde{h}_{A}$ and $(\Tilde{h}_{Y} + \Tilde{h}_{A})^{2}$ by Lemma \ref{lem::perm_GC} again. Thus, 
\begin{align*}
\mathbb{P}_{n}[\Tilde{h}_{Y} + \Tilde{h}_{A}] &= \mathbb{P}[\Tilde{h}_{Y} + \Tilde{h}_{A}] + \p1\\
    \mathbb{P}_{n}[(\Tilde{h}_{Y} + \Tilde{h}_{A})^{2}] &= \mathbb{P}[(\Tilde{h}_{Y} + \Tilde{h}_{A})^{2}] + \p1.
\end{align*}
We wish to establish $\Vert (\Tilde{h}_{Y} + \Tilde{h}_{A}) - (h_{Y} + h_{A}) \Vert_{\mathbb{P}} = \p1$. But we have already established that $\Vert \Tilde{h} - h \Vert_{\mathbb{P}} = \p1$ in the proof of Proposition 2, so by Minkowski's inequality, it is sufficient to show that $\Vert \Tilde{h}_{X} - h_{X} \Vert_{\mathbb{P}} = \p1$, which is comparatively simple. We have
\begin{align*}
    \Vert \Tilde{h}_{X} - h_{X} \Vert_{\mathbb{P}} &= \Vert \hat{\pi}(\hat{\mu}^{1} - \hat{\mu}^{0} - \hat{\psi}) - \pi(\mu^{1}-\mu^{0} - \psi_{patt})\Vert_{\mathbb{P}} \\ 
    &\leq \Vert \pi\{(\hat{\mu}^{1} - \hat{\mu}^{0} - \hat{\psi}) - (\mu^{1}-\mu^{0} - \psi_{patt})\}\Vert_{\mathbb{P}}\\
    &+ \Vert (\hat{\pi} - \pi)(\hat{\mu}^{1} - \hat{\mu}^{0} - \hat{\psi})\Vert_{\mathbb{P}}.
\end{align*}
For the first term in the upper-bound,
\begin{equation*}
    \Vert \pi\{(\hat{\mu}^{1} - \hat{\mu}^{0} - \hat{\psi}) - (\mu^{1}-\mu^{0} - \psi_{patt})\}\Vert_{\mathbb{P}} \leq \Vert \hat{\mu}^{1} - \mu^{1} \Vert_{\mathbb{P}} + \Vert \hat{\mu}^{0} - \mu^{0} \Vert_{\mathbb{P}} + |\hat{\psi} - \psi_{patt}| = \p1.
\end{equation*}
For the second term, first consider
\begin{align*}
    |\hat{\mu}^{1} - \hat{\mu}^{0} - \hat{\psi}|  &= |\hat{\mu}^{1} - Y + Y -\hat{\mu}^{0} - \hat{\psi} + \psi_{patt} - \psi_{patt}| \\ 
    &\leq 2C + |\hat{\psi} - \psi_{patt}| + \psi_{patt}
\end{align*}
with $\mathbb{P}$-probability 1 by Assumptions \ref{ass::patt_eff} and \ref{ass::prox_eff}. Hence,
\begin{equation*}
    \Vert (\hat{\pi} - \pi)(\hat{\mu}^{1} - \hat{\mu}^{0} - \hat{\psi})\Vert_{\mathbb{P}} \leq  \Vert \hat{\pi} - \pi \Vert_{\mathbb{P}}(2C + |\hat{\psi} - \psi_{patt}| + \psi_{patt}) = \p1.
\end{equation*}
Now, Lemma \ref{lem::l2} implies that $\mathbb{P}[(\Tilde{h}_{Y} + \Tilde{h}_{A})^{2}] = \mathbb{P}[(h_{Y} + h_{A})^{2}] + \p1$ and $\mathbb{P}[\Tilde{h}_{Y} + \Tilde{h}_{A}] = \mathbb{P}[h_{Y} + h_{A})] + \p1 = \p1$.
Finally, we have
\begin{equation*}
    \hat{V}_{actt} = \frac{\mathbb{P}_{n}[(\Tilde{h}_{Y} + \Tilde{h}_{A})^{2}] - \mathbb{P}_{n}[\Tilde{h}_{Y} + \Tilde{h}_{A}]^{2}}{\mathbb{P}_{n}[A]^{2}} = \frac{\mathbb{P}[(h_{Y}+h_{A})^{2}]}{\mathbb{P}[A]^{2}} + \p1 =\Vert \dot{\psi}^{Y} + \dot{\psi}^{A}\Vert_{\mathbb{P}}^{2} + \p1.
\end{equation*}

\subsection{Proof of Theorem 1}

Using Proposition 1, we can start with $n^{1/2}(\hat{\psi} - \psi_{patt}) = \mathbb{G}_{n}[\dot{\psi}] + o_{\mathbb{P}}(1)$. Then we have
\begin{equation*}
    n^{1/2}(\hat{\psi} - \psi_{swatt}) = \mathbb{G}_{n}[\dot{\psi}] -\frac{n^{1/2}\mathbb{P}_{n}[\pi(Y^{1} - Y^{0} - \psi_{patt})]}{\mathbb{P}_{n}[\pi]} + \p1.
\end{equation*}
Consider the numerator in the second term on the right. First, the expectation with respect to $\mathbb{P}$ is
\begin{align*}
    \mathbb{P}[\pi(Y^{1}-Y^{0} - \psi_{patt})] &= \mathbb{P}[\pi(\mu^{1} - \mu^{0})] - \mathbb{P}[A]\psi_{patt} \\
    &= \mathbb{P}[A(Y^{1} - Y^{0})] - \mathbb{P}[A]\psi_{patt}\\
    &= 0.
\end{align*}
The first equality is due to the tower property with conditioning on $X$. The second equality uses the tower property in reverse combined with the ignorability in Assumption \ref{ass::strong_ign_pos}. Next, the variance is also finite:
\begin{align*}
    \Vert\pi(Y^{1}-Y^{0} - \psi_{patt})\Vert^{2}_{\mathbb{P}} &\leq 3(\Vert Y^{1} \Vert^{2}_{\mathbb{P}}+ \Vert Y^{0} \Vert^{2}_{\mathbb{P}} + \psi_{patt}^{2}) < \infty
\end{align*}
by the assumed square-integrability of $Y^{1}$ and $Y^{0}$. So the central limit theorem implies that $n^{1/2}\mathbb{P}_{n}[\pi(Y^{1} - Y^{0} - \psi_{patt})] = \Op1$ and we deduce from Lemma \ref{lem::swap} that
\begin{equation*}
    \frac{n^{1/2}\mathbb{P}_{n}[\pi(Y^{1} - Y^{0} - \psi_{patt})]}{\mathbb{P}_{n}[\pi]} = \frac{n^{1/2}\mathbb{P}_{n}[\pi(Y^{1} - Y^{0} - \psi_{patt})]}{\mathbb{P}[A]} + \p1.
\end{equation*}
Putting everything together yields
\begin{align*}
    n^{1/2}(\hat{\psi} - \psi_{swatt}) &= \mathbb{G}_{n}[\dot{\psi}] - \frac{n^{1/2}\mathbb{P}_{n}[\pi(Y^{1} - Y^{0} - \psi_{patt})]}{\mathbb{P}[A]} + \p1 \\
    &= n^{1/2}\mathbb{P}\left[\dot{\psi}^{Y}+ \dot{\psi}^{A} - \frac{\pi\{Y^{1} - Y^{0} - \mu^{1}(X) + \mu^{0}(X)\}}{\mathbb{P}[A]}\right] + \p1.
\end{align*}
Let $h_{\pi} = \pi\{Y^{1} - Y^{0} - \mu^{1}(X) + \mu^{0}(X)\}$. Write
\begin{align*}
    \dot{\psi}^{Y} &= \frac{h_{\pi}}{\mathbb{P}[A]} + \frac{A-\pi(X)}{\mathbb{P}[A]}\left[Y^{1}-\mu^{1}(X)+\frac{\pi(X)}{1-\pi(X)}\{Y^{0} - \mu^{0}(X)\}\right],
\end{align*}
so
\begin{equation*}
    \langle \dot{\psi}^{Y}, h_{\pi}\rangle = \frac{1}{\mathbb{P}[A]}\left(\Vert h_{\pi}\Vert_{\mathbb{P}}^{2} + \left\langle A- \pi, h_{\pi}\left[Y^{1}-\mu^{1}(X)+\frac{\pi(X)}{1-\pi(X)}\{Y^{0} - \mu^{0}(X)\}\right]\right\rangle\right).
\end{equation*}
The second term on the right, expressed by the inner product, is equal to zero by applying the tower property conditioning on $(Y^{1}, Y^{0}, X)$ and the fact that $\mathbb{E}(A - \pi(X) \mid Y^{1}, Y^{0}, X) = \mathbb{E}(A - \pi(X) \mid X) = 0$ by strong ignorability. A similar argument yields $\langle \dot{\psi}^{A}, h_{\pi}\rangle = 0$. We deduce that
\begin{equation*}
    \left\Vert\dot{\psi}^{Y}+ \dot{\psi}^{A} - \frac{\pi\{Y^{1} - Y^{0} - \mu^{1}(X) - \mu^{0}(X)\}}{\mathbb{P}[A]}\right\Vert^{2}_{\mathbb{P}} = \Vert \dot{\psi}^{Y}\Vert_{\mathbb{P}}^{2}+ \Vert \dot{\psi}^{A}\Vert_{\mathbb{P}}^{2} - \frac{\Vert h_{\pi}\Vert_{\mathbb{P}}^{2}}{\mathbb{P}[A]^{2}},
\end{equation*}
and $\Vert h_{\pi} \Vert_{\mathbb{P}}^{2} = \mathbb{E}\{\pi(X)^{2}\text{var}(Y^{1}-Y^{0} \mid X)\}$ by applying the tower property conditioning on $X$.

\subsection{Proof of Proposition 4}
Given Assumptions \ref{ass::patt_eff} and \ref{ass::cond_var}, we can apply Lemma \ref{lem::perm_GC} to deduce that $\hat{\pi}^{2}(\hat{\sigma}_{1} - \hat{\sigma}_{0})^{2}$ takes values in a fixed $\mathbb{P}$-Glivenko-Cantelli class, so
\begin{equation*}
    \mathbb{P}_{n}[\hat{\pi}^{2}(\hat{\sigma}_{1} - \hat{\sigma}_{0})^{2}] = \mathbb{P}[\hat{\pi}^{2}(\hat{\sigma}_{1} - \hat{\sigma}_{0})^{2}] + \p1.
\end{equation*}

By Lemmas \ref{lem::swap} and \ref{lem::l2}, it now suffices to show that $\hat{\pi}(\hat{\sigma}_{1} - \hat{\sigma}_{0})$ converges in $L_{2}(\mathbb{P})$ to $\pi(\sigma_{1} - \sigma_{0})$:
\begin{align*}
    \Vert \hat{\pi}(\hat{\sigma}_{1} - \hat{\sigma}_{0}) - \pi(\sigma_{1} - \sigma_{0})\Vert_{\mathbb{P}} &\leq \Vert (\hat{\pi} - \pi)(\hat{\sigma}_{1} - \hat{\sigma}_{0})\Vert_{\mathbb{P}} + \Vert \pi\{(\hat{\sigma}_{1} - \hat{\sigma}_{0}) - (\sigma_{1} - \sigma_{0})\}\Vert_{\mathbb{P}} \\
    &\lesssim \Vert \hat{\pi} - \pi \Vert + \Vert \hat{\sigma}_{1} - \sigma_{1} \Vert_{\mathbb{P}} + \Vert \hat{\sigma}_{0} - \sigma_{0} \Vert_{\mathbb{P}} \\
    &= \p1,
\end{align*}
where we have used the uniform boundedness of the classes for $\hat{\sigma}_{1}$ and $\hat{\sigma}_{0}$ for the inequality on the second line.

\subsection{Proof of Proposition 5}

Using Proposition 1, we can start with $n^{1/2}(\hat{\psi} - \psi_{patt}) = n^{1/2}\mathbb{P}_{n}[\dot{\psi}] + o_{\mathbb{P}}(1)$, so
\begin{equation*}
    n^{1/2}(\hat{\psi} - \psi_{catt}) = n^{1/2}(\psi_{patt} + \mathbb{P}_{n}[\dot{\psi}^{Y}] + \mathbb{P}_{n}[\dot{\psi}^{A}+\dot{\psi}^{X}] - \psi_{catt}) + \p1.
\end{equation*}
Thus, it suffices to show that $n^{1/2}(\psi_{patt} + \mathbb{P}_{n}[\dot{\psi}^{A}+\dot{\psi}^{X}] - \psi_{catt}) = \p1$ since the remaining term $n^{1/2}\mathbb{P}_{n}[\dot{\psi}^{Y}]$ converges weakly to the desired normal distribution. We can write
\begin{equation*}
    n^{1/2}(\psi_{patt} + \mathbb{P}_{n}[\dot{\psi}^{A}+\dot{\psi}^{X}] - \psi_{catt}) = n^{1/2} \mathbb{P}_{n}[A\{\mu^{1}-\mu^{0} - \psi_{patt}\}]\left(\frac{1}{\mathbb{P}[A]} - \frac{1}{\mathbb{P}_{n}[A]}\right).
\end{equation*}
Since $A\{\mu^{1}-\mu^{0} - \psi_{patt}\}$ has mean zero and finite variance, the central limit theorem implies that $n^{1/2} \mathbb{P}_{n}[A\{\mu^{1}-\mu^{0} - \psi_{patt}\}] = \Op1$. Then we apply Lemma \ref{lem::swap} to deduce the result.

For the variance estimation, we use the same notation as the proof of Proposition 3. Arguing similarly to the proof of Proposition 3, we want to establish
\begin{align*}
\mathbb{P}_{n}[\Tilde{h}_{Y}] &= \mathbb{P}[h_{Y}] + \p1 = \p1\\
    \mathbb{P}_{n}[\Tilde{h}_{Y}^{2}] &= \mathbb{P}[h_{Y}^{2}] + \p1,
\end{align*}
and it is sufficient to show that $\Vert \Tilde{h}_{A} + \Tilde{h}_{X} - h_{A} - h_{X}\Vert_{\mathbb{P}} = \p1$. Then
\begin{align*}
    \Vert \Tilde{h}_{A} + \Tilde{h}_{X} - h_{A} - h_{X}\Vert_{\mathbb{P}} &= \Vert A\{(\hat{\mu}^{1} - \hat{\mu}^{0} - \hat{\psi}) - (\mu^{1}-\mu^{0} - \psi_{patt})\}\Vert_{\mathbb{P}}\\
    &\leq \Vert \hat{\mu}^{1} - \mu^{1} \Vert_{\mathbb{P}} + \Vert \hat{\mu}^{0} - \mu^{0} \Vert_{\mathbb{P}} + |\hat{\psi} - \psi_{patt}|\\
    &= \p1.
\end{align*}
Putting the ingredients together with Lemma \ref{lem::swap} yields
\begin{equation*}
    \hat{V}_{catt} = \frac{\mathbb{P}_{n}[\Tilde{h}_{Y}^{2}] - \mathbb{P}_{n}[\Tilde{h}_{Y}]^{2}}{\mathbb{P}_{n}[A]^{2}} = \frac{\mathbb{P}[h_{Y}^{2}]}{\mathbb{P}[A]^{2}} + \p1 =\Vert \dot{\psi}^{Y}\Vert_{\mathbb{P}}^{2} + \p1.
\end{equation*}

\subsection{Proof of Theorem 2}

Using Proposition 1, we can start with $n^{1/2}(\hat{\psi} - \psi_{patt}) = n^{1/2}\mathbb{P}_{n}[\dot{\psi}] + o_{\mathbb{P}}(1)$, so
\begin{align*}
    n^{1/2}(\hat{\psi} - \psi_{satt}) &= n^{1/2}(\psi_{patt} + \mathbb{P}_{n}[\dot{\psi}] - \psi_{satt}) + \p1\\
    &= \frac{n^{1/2}}{\mathbb{P}_{n}[A]}\mathbb{P}_{n}\left[\frac{A - \pi(X)}{\{1-\pi(X)\}}\{Y - \mu^{0}(X)\} - A(Y-Y^{0})\right]+\p1\\
    &= \frac{n^{1/2}}{\mathbb{P}_{n}[A]}\mathbb{P}_{n}\left[\frac{\{Y - \mu^{0}(X)\}(1-A)\pi(X)}{\{1-\pi(X)\}} - A\{\mu^{0}(X)-Y^{0}\}\right]+ \p1\\
    &= \frac{n^{1/2}}{\mathbb{P}_{n}[A]}\mathbb{P}_{n}[\mathbb{P}[A]\dot{\tau}^{Y} - A\{\mu^{0}(X)-Y^{0}\}]+\p1.
\end{align*}
It is clear that $\langle \dot{\tau}^{Y}, A\{\mu^{0}(X)-Y^{0}\}\rangle = 0$ because $A(1-A) = 0$ with $\mathbb{P}$-probability 1. Therefore, $n^{1/2}\mathbb{P}_{n}[\mathbb{P}[A]\dot{\tau}^{Y} - A\{\mu^{0}(X)-Y^{0}\}]$ converges weakly to $\mathcal{N}(0, \mathbb{P}_{n}[A]^{2}\Vert \dot{\tau}^{Y} \Vert_{\mathbb{P}}^{2} + \text{var}\{A(Y^{0} - \mu^{0})\})$. Then we apply Lemma \ref{lem::swap} and Slutsky's lemma to obtain the required limiting normal distribution.

For the variance estimation, we begin by showing that we can consistently estimate $\Vert \dot{\tau}^{Y} \Vert_{\mathbb{P}}^{2}$. Consider
\begin{align*}
    \frac{\hat{\pi}(1-A)}{1-\hat{\pi}}(Y^{0}- \hat{\mu}^{0}) - \frac{\pi(1-A)}{1-\pi}(Y - \mu^{0})&= (\mu^{0} - \hat{\mu}^{0})\frac{\pi(1-A)}{1-\pi}\\
    &+ (1-A)\left\{\frac{\hat{\pi}}{1-\hat{\pi}} - \frac{\pi}{1-\pi}\right\}(Y - \hat{\mu}^{0})\\
    &= (\mu^{0} - \hat{\mu}^{0})\frac{\pi(1-A)}{1-\pi}\\
    &+ (1-A)\frac{\hat{\pi} - \pi}{(1-\hat{\pi})(1-\pi)}(Y - \hat{\mu}^{0}).
\end{align*}
Thus, 
\begin{equation*}
    \left\Vert \frac{\hat{\pi}(1-A)}{1-\hat{\pi}}(Y- \hat{\mu}^{0}) - \frac{\pi(1-A)}{1-\pi}(Y - \mu^{0})\right\Vert_{\mathbb{P}} \leq \frac{1}{\delta} \Vert \hat{\mu}^{0} - \mu^{0} \Vert_{\mathbb{P}} + \frac{C}{\delta^{2}}\Vert \hat{\pi} - \pi \Vert_{\mathbb{P}} = \p1.
\end{equation*}
By applying similar arguments to the proof of Proposition 2, we obtain
\begin{equation*}
    \mathbb{P}_{n}\left[\frac{(1-A)\hat{\pi}(X)^{2}}{\{1-\hat{\pi}(X)\}^{2}}\left\{\frac{Y - \hat{\mu}^{0}(X)}{\mathbb{P}_{n}[A]}^{2}\right\}\right] \rightarrow \Vert \dot{\tau}^{Y}\Vert_{\mathbb{P}}^{2}
\end{equation*}
in $\mathbb{P}$-probability.

Next we show that $\text{var}\{A(Y^{0} - \mu^{0})\}$ is identified. By the tower property and ignorability, we have
\begin{equation*}
    \mathbb{E}\{A(Y^{0} - \mu^{0})\} = \mathbb{E}[\mathbb{E}\{A(Y^{0} - \mu^{0}) \mid X\}] = \mathbb{E}[\pi\mu^{0} - \pi\mu^{0}] = 0.
\end{equation*}
Then
\begin{align*}
    \text{var}\{A(Y^{0} - \mu^{0})\} &= \mathbb{E}\{A^{2}(Y^{0} - \mu^{0})^{2}\}\\
    &= \mathbb{E}[\mathbb{E}\{A(Y^{0}-\mu^{0})^{2} \mid X\}]\quad \text{(tower property)}\\
    &= \mathbb{E}\{\pi(Y^{0} - \mu^{0})^{2}\}\quad \text{(ignorability)}\\
    &= \mathbb{E}\left[\frac{\pi(1-A)}{1-\pi}(Y^{0} - \mu^{0})^{2}\right]\quad \text{(ignorability)}\\
    &= \mathbb{E}\left[\frac{\pi(1-A)}{1-\pi}(Y - \mu^{0})^{2}\right].
\end{align*}
Now we want to prove that
\begin{equation*}
    \mathbb{P}\left[\frac{\hat{\pi}(1-A)}{1-\hat{\pi}}(Y- \hat{\mu}^{0})^{2} - \frac{\pi(1-A)}{1-\pi}(Y - \mu^{0})^{2}\right] = \p1.
\end{equation*}
We will do this by leveraging the earlier computations. Consider
\begin{align*}
    \frac{\hat{\pi}(1-A)}{1-\hat{\pi}}(Y- \hat{\mu}^{0})^{2} - \frac{\pi(1-A)}{1-\pi}(Y - \mu^{0})^{2}&= \left(\frac{1-\hat{\pi}}{\hat{\pi}}\right)\frac{\hat{\pi}^{2}(1-A)}{(1-\hat{\pi})^{2}}(Y- \hat{\mu}^{0})^{2} - \left(\frac{1-\pi}{\pi}\right)\frac{\pi^{2}(1-A)}{(1-\pi)^{2}}(Y - \mu^{0})^{2}\\
    &= \left(\frac{1-\pi}{\pi}\right)\left\{\frac{\hat{\pi}^{2}(1-A)}{(1-\hat{\pi})^{2}}(Y- \hat{\mu}^{0})^{2} - \frac{\pi^{2}(1-A)}{(1-\pi)^{2}}(Y - \mu^{0})^{2}\right\}\\
    &+ \left(\frac{\pi - \hat{\pi}}{\pi\hat{\pi}}\right)\frac{\hat{\pi}^{2}(1-A)}{(1-\hat{\pi})^{2}}(Y - \hat{\mu}^{0})^{2}.
\end{align*}
Then we have
\begin{align*}
    \mathbb{P}\left[\left|\frac{\hat{\pi}(1-A)}{1-\hat{\pi}}(Y- \hat{\mu}^{0})^{2} - \frac{\pi(1-A)}{1-\pi}(Y - \mu^{0})^{2}\right|\right]&\leq \frac{1}{\delta} \mathbb{P}\left[\left|\frac{\hat{\pi}^{2}(1-A)}{(1-\hat{\pi})^{2}}(Y- \hat{\mu}^{0})^{2} - \frac{\pi^{2}(1-A)}{(1-\pi)^{2}}(Y - \mu^{0})^{2}\right|\right]\\
    &+ \frac{C^{2}}{\delta^{3}}\mathbb{P}[|\pi - \hat{\pi}|].
\end{align*}
The first term on the right is $\p1$ by using the $L_{2}$-convergence established earlier in the proof along with Lemma \ref{lem::l2}. The second term on the right is $\p1$ by the assumed $L_{2}$ convergence of $\hat{\pi}$ and the Cauchy-Schwarz inequality. Moreover, 
\begin{equation*}
    \frac{\hat{\pi}(1-A)}{1-\hat{\pi}}(Y- \hat{\mu}^{0})^{2}
\end{equation*}
takes values in a fixed $\mathbb{P}$-Glivenko-Cantelli class by applying Lemma \ref{lem::perm_GC} and Assumption \ref{ass::prox_eff}. Then by similar arguments to the proof of Proposition 2, 
\begin{equation*}
    \mathbb{P}_{n}\left[\frac{(1-A)\hat{\pi}(X)}{\{1-\hat{\pi}(X)\}}\left\{\frac{Y - \hat{\mu}^{0}(X)}{\mathbb{P}_{n}[A]}^{2}\right\}\right] \rightarrow \mathbb{P}[A]^{-2}\text{var}\{A(Y^{0} - \mu^{0})\}
\end{equation*}
in $\mathbb{P}$-probability. Combining this with the earlier convergence result for $\Vert \dot{\tau}^{Y} \Vert_{\mathbb{P}}^{2}$ completes the proof.

\subsection{Proof of Proposition 6}

We can relate $\psi_{matt}$ to $\psi_{satt}$ by
\begin{equation*}
    \psi_{matt} = \psi_{satt} + \frac{\mathbb{P}_{n}[A(Y^{0} - \mu^{0})]}{\mathbb{P}_{n}[A]}.
\end{equation*}
So we deduce from the proof of Theorem 2 that
\begin{equation*}
    n^{1/2}(\hat{\psi} - \psi_{matt}) = \frac{n^{1/2}\mathbb{P}[A]}{\mathbb{P}_{n}[A]}\mathbb{P}_{n}[\dot{\tau}^{Y}]+\p1.
\end{equation*}
Then the central limit theorem and Lemma \ref{lem::swap} yield the required normal weak limit. The consistency of the variance estimator also follows from the proof of Theorem 2.

For the last statement of the proposition, write
\begin{equation*}
    \dot{\psi}^{Y} = \frac{A(Y - \mu^{1})}{\mathbb{P}[A]} - \dot{\tau}^{Y}.
\end{equation*}
The two terms on the right are orthogonal because $A(1-A)=0$ $\mathbb{P}$-almost surely. Thus,
\begin{equation*}
    \Vert \dot{\psi}^{Y} \Vert_{\mathbb{P}}^{2} = \Vert \dot{\tau}^{Y} \Vert_{\mathbb{P}}^{2} + \mathbb{P}[A]^{-2}\text{var}\{A(Y - \mu^{1})\}.
\end{equation*}

\section{Auxiliary lemmas}

\begin{lemma}
    \label{lem::swap}
    Suppose Assumption \ref{ass::prox_eff} holds. Then
    \begin{align*}
        \frac{1}{\mathbb{P}_{n}[A]} - \frac{1}{\mathbb{P}[A]} &= O_{\mathbb{P}}(n^{-1/2})\\
        \frac{1}{\mathbb{P}_{n}[\pi]} - \frac{1}{\mathbb{P}[A]} &= O_{\mathbb{P}}(n^{-1/2}).
    \end{align*}
\end{lemma}
\begin{proof}
    By Assumption \ref{ass::prox_eff}, we have $\mathbb{P}[A] > 0$. Write
    \begin{equation} \label{eqn::sqrt_diff}
        n^{1/2}\left(\frac{1}{\mathbb{P}_{n}[A]} - \frac{1}{\mathbb{P}[A]}\right) = \frac{n^{1/2}(\mathbb{P}[A] - \mathbb{P}_{n}[A])}{\mathbb{P}_{n}[A]\mathbb{P}[A]}.
    \end{equation}
    By the central limit theorem, the numerator $n^{1/2}(\mathbb{P}[A] - \mathbb{P}_{n}[A])$ converges weakly to $\mathcal{N}(0, \mathbb{P}[A](1-\mathbb{P}[A]))$. So Slutsky's lemma implies that (\ref{eqn::sqrt_diff}) converges weakly to $\mathcal{N}(0, \mathbb{P}[A]^{-3}(1-\mathbb{P}[A]))$, which is $O_{\mathbb{P}}(1)$. A similar argument yields the second expression.
\end{proof}

\begin{lemma}
    \label{lem::l2}
    Suppose that $\hat{f}$ is a sequence of random functions taking values in $L_{2}(\mathbb{P})$ such that $\Vert \hat{f} - f \Vert_{\mathbb{P}} = \p1$ for some $f \in L_{2}(\mathbb{P})$. Then 
    \begin{align*}
        \mathbb{P}[|\hat{f}^{2} - f^{2}|] &= \p1 \\
        \mathbb{P}[|\hat{f} - f|] &= \p1.
    \end{align*}
\end{lemma}
\begin{proof}
    For the first objective, we have
    \begin{align*}
        \mathbb{P}[|\hat{f}^{2} - f^{2}|]
        &= \mathbb{P}[|(\hat{f} - f)^{2} - 2f(f - \hat{f})|] \\
        & \leq \Vert \hat{f} - f \Vert_{\mathbb{P}} + 2\Vert f \Vert_{\mathbb{P}} \Vert \hat{f} - f \Vert_{\mathbb{P}}\\
        &= \p1,
    \end{align*}
    where the inequality on the penultimate line used the Minkowski and Cauchy-Schwarz inequalities. For the second objective, we have
    \begin{align*}
        \leq \mathbb{P}[|\hat{f} - f|]
        &\leq \Vert \hat{f} - f \Vert_{\mathbb{P}}\\
        &= \p1
    \end{align*}
    by Cauchy-Schwarz again.
\end{proof}

The remaining lemmas are well-known empirical process results that we state here for convenient reference.

\begin{lemma}[Theorem 2.10.5 of \citet{vanderVaart23}] \label{lem::perm_GC}
Let $\mathcal{F}_{1},\ldots, \mathcal{F}_{k}$ be $\mathbb{P}$-Glivenko-Cantelli classes with integrable envelopes. If $\phi:\mathbb{R}^{k}\rightarrow \mathbb{R}$ is continuous, then $\phi \circ (\mathcal{F}_{1},\ldots, \mathcal{F}_{k})$ is $\mathbb{P}$-Glivenko-Cantelli provided that it has an integrable envelope function.
\end{lemma}

\begin{lemma}[Examples 2.10.9, 2.10.10 and 2.10.11 of \citet{vanderVaart23}] \label{lem::perm_Dons}
Suppose $\mathcal{F}$ and $\mathcal{G}$ are $\mathbb{P}$-Donsker classes. 
\begin{itemize}
    \item[(i)] If $\sup_{f \in \mathcal{F} \cup \mathcal{G}} |\mathbb{P}[f]| < \infty$, then the set of pairwise sums $\mathcal{F} + \mathcal{G}$ is a $\mathbb{P}$-Donsker class.
    \item[(ii)] If $\mathcal{F}$ and $\mathcal{G}$ are uniformly bounded, then the set of pairwise products $\mathcal{F}\mathcal{G}$ is a $\mathbb{P}$-Donsker class.
    \item[(iii)] If $\sup_{f \in \mathcal{F}}|\mathbb{P}[f]| < \infty$ and $f \geq \delta$ for some $\delta > 0$ for every $f \in \mathcal{F}$, then $1/\mathcal{F} = \{1/f: f\in \mathcal{F}\}$ is a $\mathbb{P}$-Donsker class.
\end{itemize}
    
\end{lemma}

\section{The sharpest bound for the asymptotic variance of the sample weighted effect}

Recall that the asymptotic variance of $\psi_{swatt}$ is equal to 
\begin{equation*}
    \Vert \dot{\psi}^{Y} + \dot{\psi}^{A}\Vert_{\mathbb{P}}^{2} - \mathbb{P}[A]^{-2}\mathbb{E}\{\pi(X)^{2}\text{var}(Y^{1}-Y^{0} \mid X)\},
\end{equation*}
as shown in Theorem 1. We mentioned that one could simply estimate $\Vert\dot{\psi}^{Y} + \dot{\psi}^{A}\Vert_{\mathbb{P}}^{2}$, the asymptotic variance of $\psi_{actt}$, to use as a simple asymptotically conservative variance estimator. This has the benefit of not requiring any additional estimation of nuisance parameters. Alternatively, we showed using Cauchy-Schwarz that 
\begin{equation*}
    \mathbb{E}[\pi(X)^{2}\{\sigma_{1}(X) - \sigma_{0}(X)\}^{2}] \leq \mathbb{E}\{\pi(X)^{2}\text{var}(Y^{1}-Y^{0} \mid X)\},
\end{equation*}
which gives us a sharper bound on the asymptotic variance but requires conditional variance estimation to operationalize.

The sharpest possible bound on the asymptotic variance of $\psi_{swatt}$ is obtained through the Fr\'echet-Hoeffding upper bound. For each $a \in \{0,1\}$, let $F_{a \mid x}$ be the distribution function of $Y^{a}$ given $X=x$. These are identified by
\begin{equation*}
    F_{a \mid x}(y) = \mathbb{P}(Y^{a} \leq y \mid X=x) = \mathbb{P}(Y \leq y \mid A=a, X=x).
\end{equation*}
Then the Fr\'echet-Hoeffding upper bound for the conditional covariance is
\begin{equation*}
    \text{cov}^{H}(x) = \int_{0}^{1} F_{1\mid x}^{-1}(u)F_{0\mid x}^{-1}(u)du -\mu^{1}(x)\mu^{0}(x) \geq \text{cov}(Y^{1},Y^{0} \mid X=x),
\end{equation*}
which is attained if $Y^{1}$ and $Y^{0}$ are comonotonic given $X=x$; that is, $(Y^{1}, Y^{0}) \mid X=x \sim (F_{1\mid x}^{-1}(U), F_{0\mid x}^{-1}(U))$ for $U \sim U[0,1]$. Consequently,
\begin{equation*}
    \mathbb{E}[\pi(X)^{2}\{\sigma_{1}^{2}(X) + \sigma_{0}^{2}(X) - 2\text{cov}^{H}(X)\}]\leq \mathbb{E}\{\pi(X)^{2}\text{var}(Y^{1}-Y^{0} \mid X)\}.
\end{equation*}
This suggests that we should estimate $\text{cov}^{H}$ to obtain the sharpest possible asymptotic variance estimate for $\psi_{swatt}$. However, this generally appears to require simultaneous quantile regression to obtain plug-in estimates $\hat{F}_{1\mid x}^{-1}(u)$ and $\hat{F}_{0\mid x}^{-1}(u)$ across all $u \in [0,1]$ and all values of $x$. We defer the development of this methodology to future work. An exception is the case of binary outcomes, where the sharpest bound takes a particularly simple form.
\setcounter{proposition}{6}
\setcounter{example}{1}
\begin{proposition}
    For binary outcomes, the lower bound
    \begin{equation*}
        \mathbb{E}[\pi(X)^{2}\{|\mu^{1}(X) - \mu^{0}(X)| - |\mu^{1}(X) - \mu^{0}(X)|^{2}\}] \leq \mathbb{E}\{\pi(X)^{2}\text{var}(Y^{1}-Y^{0} \mid X)\}
    \end{equation*}
    is sharp. Under Assumptions \ref{ass::patt_eff} and \ref{ass::prox_eff}, the lower bound is consistently estimated by
    \begin{equation*}
        \hat{V}_{FH} = \mathbb{P}_{n}[\hat{\pi}(X)^{2}\{|\hat{\mu}^{1}(X) - \hat{\mu}^{0}(X)| - |\hat{\mu}^{1}(X) - \hat{\mu}^{0}(X)|^{2}\}].
    \end{equation*}
\end{proposition}
\begin{proof}
    The conditional covariance between $Y^{1}$ and $Y^{0}$ can be written as
    \begin{equation*}
        \text{cov}(Y^{1},Y^{0} \mid X=x) = \mathbb{E}(Y^{1}Y^{0} \mid X=x) - \mu^{1}(x)\mu^{0}(x).
    \end{equation*}
    The first term on the right is bounded above by
    \begin{equation*}
        \mathbb{E}(Y^{1}Y^{0} \mid X=x) = \mathbb{P}(Y^{1} = 1, Y^{0} = 1 \mid X=x) \leq \mu^{1}(x) \wedge \mu^{0}(x).
    \end{equation*}
    This upper bound can always be attained if the only restriction on the joint conditional distribution $(Y^{1}, Y^{0}) \mid X=x$ is the pair $(\mu^{1}(x),\mu^{0}(x))$. Hence,
    \begin{align*}
        \text{var}(Y^{1}-Y^{0} \mid X=x) &\geq \mu^{1}(x)\{1-\mu^{1}(x)\}+\mu^{0}(x)\{1-\mu^{0}(x)\}-2\{\mu^{1}(x) \wedge \mu^{0}(x) - \mu^{1}(x)\mu^{0}(x)\}\\
        &= -\{\mu^{1}(x)-\mu^{0}(x)\}^{2} + \mu^{1}(x)+\mu^{0}(x) - 2\mu^{1}(x) \wedge \mu^{0}(x)\\
        &= -\{\mu^{1}(x)-\mu^{0}(x)\}^{2} + |\mu^{1}(x)-\mu^{0}(x)|.
    \end{align*}
    Plugging this into the expectation completes the proof of the first part.

    Using similar arguments to the proofs of previous propositions, the assumptions imply that $\hat{\pi}^{2}\{|\hat{\mu}^{1} - \hat{\mu}^{0}| - |\hat{\mu}^{1} - \hat{\mu}^{0}|^{2}\}$ takes values in a fixed $\mathbb{P}$-Glivenko-Cantelli class, so
    \begin{equation*}
        \hat{V}_{FH} = \mathbb{P}[\hat{\pi}(X)^{2}\{|\hat{\mu}^{1}(X) - \hat{\mu}^{0}(X)| - |\hat{\mu}^{1}(X) - \hat{\mu}^{0}(X)|^{2}\}] + o_{\mathbb{P}}(1).
    \end{equation*}
    Moreover, the argument in the proof of Proposition 4 can be used to show
    \begin{equation*}
        \Vert \hat{\pi}(\hat{\mu}^{1} - \hat{\mu}^{0}) - \pi(\mu^{1} - \mu^{0})\Vert_{\mathbb{P}} = o_{\mathbb{P}}(1)
    \end{equation*}
    after replacing $(\hat{\sigma}_{1},\hat{\sigma}_{0}, \sigma_{1}, \sigma_{0})$ with $(\hat{\mu}^{1},\hat{\mu}^{0}, \mu^{1}, \mu^{0})$. Using Lemma \ref{lem::l2}, we have
    \begin{align*}
        \mathbb{P}[\hat{\pi}(X)^{2}|\hat{\mu}^{1}(X) - \hat{\mu}^{0}(X)|^{2}] &= \mathbb{P}[\pi(X)^{2}|\mu^{1}(X) - \mu^{0}(X)|^{2}] + o_{\mathbb{P}}(1)\\
        \mathbb{P}[\hat{\pi}(X)^{2}|\hat{\mu}^{1}(X) - \hat{\mu}^{0}(X)|] &= \mathbb{P}[\hat{\pi}(X)\pi(X)|\mu^{1}(X) - \mu^{0}(X)|]+ o_{\mathbb{P}}(1)\\
        &=\mathbb{P}[\pi(X)^{2}|\mu^{1}(X) - \mu^{0}(X)|] + o_{\mathbb{P}}(1).
    \end{align*}
    Combining the above yields the result.
\end{proof}
Consequently, we can use $n^{-1}(\hat{V}_{actt} - \mathbb{P}_{n}(A)^{-1}\hat{V}_{FH})$ as the asymptotically sharp variance estimator for $\psi_{swatt}$.

\section{Additional comparisons and estimands}

\subsection{Comparisons between the sample weighted effect and the literal estimands}

In Figure 1 of the main text, it is indicated that there is no ordering between $\psi_{swatt}$ and any of $\{\psi_{catt}, \psi_{matt},\psi_{satt}\}$. We justify this assertion with the following examples.

\begin{example}[Asymptotic variance of $\psi_{swatt}$ $<$ asymptotic variance of $\psi_{matt}$]
    For arbitrary $\mu^{1}$, set $\mu^{1} - \mu^{0}$ to be an arbitrary constant, so $\dot{\psi}^{A} = 0$ $\mathbb{P}$-almost surely. Also let $Y^{1} - \mu^{1}(X)$ have any distribution such that $\text{var}(Y^{1} \mid X) > 0$ on a set of positive $\mathbb{P}$-probability and set $Y^{0} - \mu^{0}(X) = \mu^{1}(X) - Y^{1}$. Then
    \begin{equation*}
        \mathbb{E}\{\pi(X)^{2}\text{var}(Y^{1}-Y^{0} \mid X)\} = 4\mathbb{E}\{\pi(X)^{2}\text{var}(Y^{1}\mid X)\},
    \end{equation*}
    and
    \begin{align*}
        \Vert \dot{\psi}^{Y} + \dot{\psi}^{A}\Vert_{\mathbb{P}}^{2} - \mathbb{P}[A]^{-2}\mathbb{E}\{\pi(X)^{2}\text{var}(Y^{1}-Y^{0} \mid X)\} &= \mathbb{P}[A]^{-2}\mathbb{P}\left[\left\{\pi(X) + \frac{\pi(X)^{2}}{1-\pi(X)}-4\pi(X)^{2}\right\}\text{var}(Y^{1}\mid X)\right] \\
        &= \mathbb{P}[A]^{-2}\mathbb{P}\left[\frac{\pi(X)\{1-2\pi(X)\}^{2}}{1-\pi(X)}\text{var}(Y^{1}\mid X)\right].
    \end{align*}
    Compare this with the asymptotic variance of $\psi_{matt}$:
    \begin{equation*}
        \Vert \dot{\tau}^{Y}\Vert_{\mathbb{P}}^{2} = \mathbb{P}[A]^{-2}\mathbb{P}\left[\frac{\pi(X)^{2}}{1-\pi(X)}\text{var}(Y^{1}\mid X)\right].
    \end{equation*}
    A sufficient condition for the asymptotic variance of $\psi_{swatt}$ to be strictly smaller is:
    \begin{equation*}
        \{1-2\pi(X)\}^{2} < \pi(X)
    \end{equation*}
    with $\mathbb{P}$-probability 1. This is equivalent to $\pi(X) > 1/4$ with $\mathbb{P}$-probability 1.
\end{example}

\begin{example}[Asymptotic variance of $\psi_{swatt}$ $>$ asymptotic variances of $\psi_{catt}$ and $\psi_{satt}$]
    Suppose $\mu^{1},\mu^{0}$ are such that $\mathbb{E}[(\mu^{1}-\mu^{0}-\psi_{patt})^{2}] > 0$, i.e. the CATE function is not $\mathbb{P}$-almost surely constant. Also let $Y^{1} - Y^{0} = \mu^{1}(X) - \mu^{0}(X)$, so the asymptotic variance of $\psi_{swatt}$ is $\Vert \dot{\psi}^{Y} + \dot{\psi}^{A}\Vert_{\mathbb{P}}^{2}$ while the asymptotic variance of $\psi_{catt}$ is just $\Vert \dot{\psi}^{Y}\Vert_{\mathbb{P}}^{2}$. Then
    \begin{align*}
        \Vert \dot{\psi}^{A}\Vert^{2}_{\mathbb{P}} &= \mathbb{P}[A]^{-2}\mathbb{E}[(A-\pi)^{2}(\mu^{1}-\mu^{0}-\psi_{patt})^{2}]\\
        &= \mathbb{P}[A]^{-2}\mathbb{E}[\mathbb{E}\{(A-\pi)^{2}\mid X\}(\mu^{1}-\mu^{0}-\psi_{patt})^{2}]\\
        &= \mathbb{P}[A]^{-2}\mathbb{E}[\pi(1-\pi)(\mu^{1}-\mu^{0}-\psi_{patt})^{2}]\\
        &\geq \delta^{2} \mathbb{P}[A]^{-2}\mathbb{E}[(\mu^{1}-\mu^{0}-\psi_{patt})^{2}]\\
        &> 0.
    \end{align*}
    Thus, the asymptotic variance of $\psi_{catt}$ is strictly smaller than that of $\psi_{swatt}$.

    If we further let $Y^{1} = \mu^{1}(X)$ and $Y^{0} = \mu^{0}(X)$, then the asymptotic variance of $\psi_{satt}$ is also just $\Vert \dot{\psi}^{Y}\Vert_{\mathbb{P}}^{2}$, and we deduce from the above that $\psi_{satt}$ has strictly smaller asymptotic variance as well.
\end{example}

\subsection{Additional estimands}

Recall from the main text that the population effect can be written as
\begin{equation*}
    \psi_{patt} = \mathbb{E}(Y \mid A = 1) - \tau,
\end{equation*}
where $\tau = \mathbb{E}(Y^{0} \mid A=1) = \mathbb{E}\{\mu^{0}(X) \mid A = 1\}$. Sample variants could be created by replacing the first term $\mathbb{E}(Y \mid A = 1)$ by any of
\begin{equation*}
    \left\{\frac{\mathbb{P}_{n}[\pi\mu^{1}]}{\mathbb{P}_{n}[\pi]},\frac{\mathbb{P}_{n}[\pi Y^{1}]}{\mathbb{P}_{n}[\pi]},\frac{\mathbb{P}_{n}[A\mu^{1}]}{\mathbb{P}_{n}[A]},\frac{\mathbb{P}_{n}[AY]}{\mathbb{P}_{n}[A]}\right\}
\end{equation*}
and the second term $\tau$ by any of 
\begin{equation*}
    \left\{\frac{\mathbb{P}_{n}[\pi\mu^{0}]}{\mathbb{P}_{n}[\pi]},\frac{\mathbb{P}_{n}[\pi Y^{0}]}{\mathbb{P}_{n}[\pi]},\frac{\mathbb{P}_{n}[A\mu^{0}]}{\mathbb{P}_{n}[A]},\frac{\mathbb{P}_{n}[AY^{0}]}{\mathbb{P}_{n}[A]}\right\}.
\end{equation*}
This includes all of the sample variants studied in the main text; the remaining combinations appear to lack practical motivations or utilities.

We highlight two additional estimands that provide some interesting context. It is perhaps appealing to consider 
\begin{equation*}
    \psi_{\tau} = \frac{\mathbb{P}_{n}[AY]}{\mathbb{P}_{n}[A]} - \tau,
\end{equation*}
since the problem now reduces to estimating $\tau$. As expected, we have
\begin{equation*}
    n^{1/2}(\hat{\psi} - \psi_{\tau}) \rightarrow \mathcal{N}(0,\Vert \dot{\tau}\Vert_{\mathbb{P}}^{2})
\end{equation*}
using a similar analysis to before. But it is in fact possible for $\Vert \dot{\tau}\Vert_{\mathbb{P}}$ to exceed $\Vert \dot{\psi}\Vert_{\mathbb{P}}$, as illustrated by the following example.
\begin{example}
    Set $\mu^{1} - \mu^{0}$ to be an arbitrary constant, so that $\dot{\psi}^{A} = \dot{\psi}^{X} =0$ with $\mathbb{P}$-probability 1. Then \begin{align*}
        \Vert \dot{\psi}\Vert^{2}_{\mathbb{P}} &= \Vert\dot{\tau}^{Y}\Vert_{\mathbb{P}}^{2} + \mathbb{P}(A)^{-2}\text{var}\{A(Y^{1} - \mu^{1})\}\\
        \Vert \dot{\tau}\Vert^{2}_{\mathbb{P}} &= \Vert\dot{\tau}^{Y}\Vert_{\mathbb{P}}^{2} + \mathbb{P}(A)^{-2}\text{var}\{A(\mu^{0} - \tau)\}.
    \end{align*}
    So $\Vert \dot{\psi}\Vert_{\mathbb{P}} < \Vert \dot{\tau}\Vert_{\mathbb{P}}$ if and only if $\text{var}\{A(Y^{1} - \mu^{1})\} < \text{var}\{A(\mu^{0} - \tau)\}$, e.g. if $Y^{1} = \mu^{1}(X)$, then it is sufficient to have $\text{var}\{A(\mu^{0} - \tau)\} > 0$. Essentially, scenarios like the above can occur because $\mu^{1} - \mu^{0}$ may be less variable than $\mu^{0}$ on its own. 
\end{example}

Finally, it is natural to investigate the sample variant that can be estimated most precisely. This can be deduced by looking directly at the form of $\dot{\psi}$. Define
\begin{equation*}
    \Tilde{\psi} = \frac{\mathbb{P}_{n}\left[\left\{A - \frac{\pi(1-A)}{1-\pi}\right\}(Y - \mu^{0})\right]}{\mathbb{P}_{n}(A)} = \psi_{patt} + \frac{\mathbb{P}(A)}{\mathbb{P}_{n}(A)}\mathbb{P}_{n}(\dot{\psi}).
\end{equation*}
\begin{proposition}
    Under Assumption \ref{ass::patt_eff}, $n^{1/2}(\hat{\psi} - \Tilde{\psi}) \rightarrow 0$ in weak convergence.
\end{proposition}
\begin{proof}
    Using Proposition 1, we can start with $n^{1/2}(\hat{\psi} - \psi_{patt}) = n^{1/2}\mathbb{P}_{n}(\dot{\psi}) + o_{\mathbb{P}}(1)$, from which
    \begin{equation*}
        n^{1/2}(\hat{\psi} - \Tilde{\psi}) = n^{1/2}\left\{1- \frac{\mathbb{P}(A)}{\mathbb{P}_{n}(A)}\right\}\mathbb{P}_{n}(\dot{\psi})+ o_{\mathbb{P}}(1).
    \end{equation*}
    The result follows from Lemma \ref{lem::swap}.
\end{proof}

\end{document}